\documentclass[num-refs]{wiley-article}

\usepackage{color}
\usepackage{ulem}

\usepackage{siunitx}

\papertype{Original Article}

\title{Validation of point process predictions with proper scoring rules}

\author[1]{Claudio Heinrich-Mertsching}
\author[1]{Thordis L. Thorarinsdottir}
\author[1]{Peter Guttorp}
\author[2]{Max Schneider}

\affil[1]{Norwegian Computing Center\\
Postboks 114, Blindern\newline
NO-0314 Oslo, Norway}
\affil[2]{
Department of Statistics\newline
Box 354322\newline
University of Washington\newline
Seattle, WA 98195, United States of America
}

\corraddress{Claudio Heinrich-Mertsching, Norwegian Computing Center, Postboks 114, Blindern, NO-0314 Oslo Oslo, Norway}
\corremail{claudio.heinrich-mertsching@nr.no}

\fundinginfo{We acknowledge the support of The Research Council of Norway through grant 240838
 Model selection and model verification for point processes.}

\runningauthor{C. Heinrich-Mertsching et al.}

\usepackage{geometry}
\RequirePackage[OT1]{fontenc}
\RequirePackage{amsmath}
\RequirePackage{natbib}
\RequirePackage[colorlinks,citecolor=blue,urlcolor=blue]{hyperref}

\usepackage{array}
\usepackage{booktabs}
\usepackage{amsfonts,bbm,dsfont}
\usepackage{float,graphicx,multirow,psfrag,epsf,subcaption}
\usepackage{comment}
\usepackage{todonotes}
\usepackage{lineno}

\newcolumntype{C}{@{\extracolsep{3em}}c@{\extracolsep{1em}}}%

\graphicspath{{./}}

\setcitestyle{authoryear,open={(},close={)}}

\newcommand{\wh}{\widehat}
\newcommand{\ol}{\overline}

\newcommand{\R}{\mathbbm{ R}}
\newcommand{\E}{\mathbbm{ E}}

\newcommand{\bb}{\mathbf}
\DeclareMathOperator*{\argmin}{arg\,min}

\begin{document}

\maketitle






\begin{abstract}
We introduce a class of proper scoring rules for 
evaluating spatial point process forecasts based on summary statistics. These scoring rules rely on Monte-Carlo approximations of expectations and can therefore easily be evaluated for any point process model that can be simulated. In this regard, they are more flexible than the commonly used logarithmic score and other existing proper scores for point process predictions. The scoring rules allow for evaluating the calibration of a model to specific aspects of a point process, such as its spatial distribution or tendency towards clustering. Using simulations we analyze the sensitivity of our scoring rules to different aspects of the forecasts and compare it to the logarithmic score. Applications to earthquake occurrences in northern California, USA and the spatial distribution of Pacific silver firs in Findley Lake Reserve in Washington, USA highlight the usefulness of our scores for scientific model selection. 
 \end{abstract}



\section{Introduction}

Point process methodology is applied in diverse scientific fields to model and predict earthquakes \citep{Eberhard&2012}, crime rates \citep{Mohler&2011}, urban development \citep{PourtaheriVahidi-Asl2011}, plant and cellular systems \citep{Waller&2011}, and animal colonies \citep{Edelman2012}, to name but a few examples. In prediction settings, model validation methods are needed to rank competing models according to their predictive performance. Such rankings are typically obtained by proper scoring rules \citep{GneitingRaftery2007}. A particular challenge for constructing scoring rules for point processes is that mathematical properties, such as densities or moment measures, are untractable for many point process models. This obstructs the use of scoring rules based on such properties. The common logaritmic score, for example, can only be applied to compare predictions for which densities are known, which is not the case for many Markov point processes.

In this paper, we introduce a class of proper scoring rules that can be evaluated from random draws of the predictive model. Simulation methods exist for most point process models \citep{MoellerWaagepetersen2003}, including models for which densities and other mathematical properties are unknown. Therefore, simulation-based scoring rules provide a flexible approach to rank competing point process predictions. Moreover, they allow comparison of ensemble-based point pattern predictions---frequently used in atmospheric sciences, for example for predicting lightning strikes \citep{Blouin&2016}---to fully probabilistic point process models.

The scoring rules we propose are based on estimators of summary statistics such as the intensity or Ripley's $K$-function \citep{Ripley1976, Ripley1977}. They are therefore sensitive to misspecifications of the selected summary statistic in the predictive model. In particular, scores are obtained that are sensitive to either misspecified clustering behavior or misspecifications of the intensity. Earthquake rate forecasting \citep{Schorlemmer&2018} or the spatial prediction of the distribution of species \citep{Velazquez&2016} are prominent examples of point-process-valued predictions. With these examples in mind, we focus our exposition on spatial point processes.  We analyze the introduced scores in several simulation studies and apply them to two datasets; earthquake occurrences in northern California and the spatial distribution of Pacific silver fir at Findley Lake Reserve in Washington. Despite this focus on purely spatial processes, our approach extends naturally to temporal and spatio-temporal point processes.

Our proposed construction of proper scores is conceptually simple and based on the fact that proper scores remain proper under measurable mappings. This approach reduces the complex task of finding proper scores for point processes to the much simpler task of finding mappings from the space of point patterns to a simpler space were scoring rules are readily available.  We believe this approach to generally be useful for constructing proper scoring rules for predictions taking values in highly complex spaces, such as spatial or function-valued predictions. A previous application of this principle is the construction of the variogram score for spatial predictions in \cite{ScheuererHamill2015}. For this approach it is important to employ mappings that are sensitive to certain high-level-properties of the process, in order to make the derived scores sensitive to these properties as well and, therefore, interpretable. In the context of point processes, examples for such high-level-properties are the clustering behavior or the expected number of points per area. Such properties are usually analyzed by summary statistics such as Ripley's $K$-function or the intensity of the process.
We therefore employ estimators of summary statistics as mappings.

Previously proposed model validation methods for point processes include the pixel-based residual diagnostic framework of \cite{Baddeley&2005} and \cite{Zhuang2006}, pattern transformation methods as reviewed in \cite{Clements&2011}, and diagnostics based on summary statistics such as the Ripley's $K$-function \citep{Baddeley&2000,Baddeley&2011}. Such model diagnostics typically assess how well a given model explains an observed point pattern. Proper scoring rules, on the other hand, compare the performance of competing models in a prediction setting, where models are fitted out-of-sample. They can, moreover, be used as loss functions to fit parametric models, by selecting parameters of a parametric model optimizing the score in cross-validation, see \citet{Dawid&2016} and Section \ref{Earthquakes}.

In a recent publication, \citet{Brehmer&2021} give a comprehensive overview of forecast evaluation methods for point processes and put them in the context of scoring rules. They moreover introduce several proper scoring rules for point processes. The present paper complements their work, in that the scoring rules we introduce are very flexible and can be evaluated for a wider range of predictive models than the scores considered in \citet{Brehmer&2021}. In particular, our scores can be evaluated for Markov or Gibbs processes, many of which do not possess closed form expressions for summary statistics or densities. Moreover, the $K$-function score introduced here is, to the best of our knowledge, the first proper scoring rule specifically targeting point interaction. On the other hand, the evaluation of our scores relies on Monte-Carlo-approximation and therefore introduces score uncertainty and high computational costs in a trade-off: Generating a small number of Monte-Carlo samples, for example, lowers computational cost at the price of increased score uncertainty \citep{Krueger&2021}. Simpler alternatives, such as the logaritmic score or the score from \citet[Cor. 3.8]{Brehmer&2021} (cf. equation \eqref{intensityscoreBrehmer}), may thus be preferred if they can be evaluated for all predictive models.

The remainder of this article is organized as follows. Section \ref{Theory} contains the theoretical background, including brief introductions to proper scoring rules and point processes. In Section \ref{SR} we derive proper scoring rules for point processes based on summary statistics. Sections \ref{SS1} and \ref{SS2} provide simulation studies analyzing the performance of the introduced scores. The intensity score is then applied in Section \ref{Earthquakes} to earthquake predictions, where its performance can be compared to the logaritmic score. In Section \ref{Trees}, we predict the spatial distribution of Pacific silver fir by several point process models that are frequently used in ecology, and assess their performance by the $K$-function score. Section \ref{Discussion} concludes.


\section{Background and notation}\label{Theory}


\subsection{Scoring rules}

Forecasts commonly take the form of predictive distributions in order to incorporate the forecast uncertainty, or the confidence in the prediction. Scoring rules assess the accuracy of probabilistic forecasts by assigning a numerical penalty to each forecast-observation pair. Given a measurable observation space $\mathcal O$ and a set $\mathcal P$ of probability measures on $\mathcal O$, a scoring rule is a mapping
\begin{equation}\label{eq:score}
S: \mathcal O\times \mathcal P  \rightarrow \ol\R := \R\cup\{\infty\},
\end{equation}
such that the mapping $y \mapsto S(y,F)$ is integrable with respect to the measure $G$ for every $F,G \in \mathcal P$.
We assume scoring rules to be negatively oriented, interpreting the score as a penalty, such that smaller scores indicate better predictions. 

A scoring rule is {\em proper} relative to $\mathcal P$ if 
\begin{equation}\label{eq:proper}
\E_G S(Y,G) \leq \E_G S(Y,F)\quad \text{for all $F,G \in \mathcal P$,}
\end{equation}
that is, if the expected score for a random observation $Y$ with distribution $G$ is optimized if the true distribution is issued as the forecast. The scoring rule $S$ is {\em strictly proper} relative to the class $\mathcal P$ if \eqref{eq:proper} holds, with equality only if $F=G$. The concept of propriety is rooted in decision-theory and the aim of performing forecast evaluation with proper scoring rules is to encourage honesty and prevent hedging, see e.g. the discussion in Section 1 of \citet{Gneiting2011}. 

Two examples of commonly used scoring rules are the proper squared error (SE),
\begin{equation}\label{eq:MSE} S(y,F) = (y-\E_F[X])^2,
\end{equation}
and the strictly proper continuous ranked probability score (CRPS),
\begin{equation}\label{eq:CRPS}
  \text{CRPS}(y,F) := \E_F[|y-X|] -\frac 1 2\E_F[|X'-X|],
  \end{equation}
where $F$ is assumed to have finite first moment, and $X, X'$ are independent random variables with distribution $F$ \citep{Hersbach2000, LaioTamea2007, GneitingRaftery2007}. 

Competing forecasting methods can be compared by evaluating their mean scores over an out-of-sample test set, and the method with the smallest mean score is preferred. For a small set of forecast-observation pairs, the mean score is commonly associated with a large uncertainty, see \citet{ThorarinsdottirSchuhen2018}. Formal tests of the null hypothesis of equal predictive performance can also be employed, such as the Diebold-Mariano test \citep{DieboldMariano1995} or permutation tests \citep{Good2013, Moeller&2013}.

\subsection{Point processes}

A spatial point process on a bounded set $W\subset \R^2$ is a random variable $\bb X$ taking values in $W^\cup$, the space of countable subsets of $W$. We generally assume that $\bb X$ almost surely has finitely many points.
We denote by $F$ and $G$ distributions on $W^\cup$ of point processes, observed point patterns are denoted by $\bb x,\bb y,...$ while random draws from point processes are denoted by $\bb X,\bb Y,...$. 
For a function $f:W^\cup\to \R$, the notation $\E_F[f(\bb X)]$ is used to denote the expectation of $f(\bb X)$  when $\bb X$ is distributed according to $F$.
For comprehensive overviews on spatial point processes, we refer to \cite{MoellerWaagepetersen2003,DaleyVereJones2007}. 

Summary statistics of point processes are powerful tools for exploratory data analysis and model selection. 
We define summary statistics as function-valued maps from the space of point process distributions, see Definition \ref{def:SumStat} below. This definition includes the intensity function in the class of summary statistics, which is uncommon but convenient for our purposes.
Two important examples of summary statistics are the following:
\begin{example}[Intensity function]   
 The intensity function $\lambda:W\to\R$ of a point process model $F$ is the unique function (up to null sets) satisfying
\[\int_B\lambda(w)\,dw = \E_F[n(\bb X \cap B)],\]
for all measurable sets $B\subset W$. Here, $n(\bb X\cap B)$ denotes the number of points of $\bb X$ that fall into the set $B$.
\end{example}
The intensity measures the spatial distribution of points in the sense that a high intensity highlights areas where many points are expected. 
Whereas Poisson point processes are fully defined by their intensity, the intensity contains no information about interaction of points, i.e. whether the points tend to repel each other or cluster. Such interactions are targeted by second order summary statistics, an example being Ripley's $K$-function:
\begin{example}[Ripley's $K$-function]   
For a point process $F$ with intensity $\lambda,$ Ripley's $K$-function is defined as
\[K(r) = \frac{1}{|W|} \E_F\bigg[\sum_{x_1,x_2\in\bb X,\atop x_1\neq x_2} \frac{\mathbb{ 1}\{\|x_1 - x_2\|<r\}}{\lambda(x_1)\lambda(x_2)}\bigg],\]
for $r>0.$ 
\end{example}
We generally consider this generalization of the $K$-function to spatially inhomogeneous point processes, which was introduced in \cite{Baddeley&2000}.
Roughly speaking, $K(r)$ indicates clustering at distances up to $r.$  The $K$-function of a Poisson process is $K(r) = \pi r^2.$ If, for a point process model, $K(r)$ is larger than this value for small $r$, the model has more expected point pairs with distance less than $r$ than a Poisson model, and the process exhibits clustering.
Other examples of popular summary statistics include the $F$-, $G$-, and $J$-functions as well as the pair correlation function, see \citet[Chapter 4]{MoellerWaagepetersen2003}.

Bearing these examples in mind, we define a general summary statistic as follows. Summary statistics are function-valued, taking values from a space $\mathcal R$. For the intensity function, for example, we have $\mathcal R = W$ and for the $K$-function, we have $\mathcal R = (0,\infty).$ 
\begin{definition}\label{def:SumStat}
Consider a class of predictive distributions $\mathcal P$ on $W^\cup$ and a measurable space $\mathcal R$.
A {\it summary statistic} on $\mathcal P$ is a mapping $T:\mathcal P\times \mathcal R\to\R$. We sometimes denote $T_F(r)$ instead of $T(F,r)$.
A {\it summary statistic estimator} is a mapping $\wh T:W^\cup\times \mathcal R\to\R$.
\end{definition}
In particular, we assume estimators for summary statistics to be based on a single point pattern, which is the case for all standard estimators for the summary statistics mentioned above.

Finally, let us remark that throughout this paper we assume all mappings between measurable spaces to be measurable. Products of measurable spaces are equipped with the product $\sigma$-algebra. For mappings $\phi:\mathcal P\times \mathcal M\to\mathcal M'$, where $\mathcal M,\mathcal M'$ are measurable spaces and $\mathcal P$ is a space of distributions, we assume that $\phi(F,\cdot):\mathcal M\to\mathcal M'$ is measurable for all $F\in \mathcal P.$

\section{Proper scoring rules for point process predictions}\label{SR}

When dealing with forecasts taking values in a complex observation space $\mathcal O$, it may be effective to validate and compare models by focusing on a certain property of interest. This approach is not new; in the context of real- and vector-valued forecasts, the SE \eqref{eq:MSE} focuses on the mean and the Dawid-Sebastiani score \citep{DawidSebastiani1999} focuses on mean and covariance of the predictive distribution, for example. The variogram score \citep{ScheuererHamill2015} focuses on how the spatial autocorrelation of a spatial prediction decays with distance. Here, we adapt this principle and show how it can be applied to validate point process forecasts.

\subsection{Scoring based on summary statistics}

In point process prediction, researchers often focus on the number and spatial distribution of points or clustering behavior of the process. We can use summary statistics for such specific properties to construct proper scoring rules sensitive to the same property. This approach has several advantages. It is easily applicable and does not impose any conditions on the predictive distribution. Thus, it can be used to directly compare predictive performance of any collection of point process models. Secondly, the derived scoring rules are always proper, and therefore allow for easy comparison of predictive perfomance following decision-theoretic principles.

\begin{proposition}\label{unbiased}
Let $r\in \mathcal R$ be fixed. Assume that $\wh T$ is an unbiased estimator for $T$ in the sense that $\E_F[\wh T (\bb Y,r)] = T(F,r)$ for all $F\in \mathcal P$. Then the scoring rule
\[S_T(\bb y , F,r) := (\wh T(\bb y,r) - T(F,r))^2\]
is proper relative to the class of all point process models.
\end{proposition}
\begin{proof}
This follows directly from the fact that for any random variable $Y$, the function $c\mapsto \E[(Y-c)^2]$ is minimized in $c = \E[Y]$.
\end{proof}
The score $S_T$ is usually not strictly proper as we may have $T(F,r) = T(G,r)$ for distributions $F\neq G$. For example, \cite{Baddeley&1984} show that point processes with very different characteristics may have the same $K$-function.

In Proposition \ref{unbiased}, both $\wh T$ and $T$ get evaluated at a specific point $r\in\mathcal R$, whereas in practice one is usually more interested in an overall fit. To this end, we can use the following result, which is an immediate consequence of Tonelli's theorem.
\begin{proposition}\label{integral}
 If $S(\bb y, F,r)$ is a non-negative proper scoring rule relative to $\mathcal P$ for all $r\in \mathcal R$, then the scoring rule
\begin{align}\label{S_A}
S(\bb y ,F):=\int_{\mathcal R} S(\bb y ,F,r)w(r)dr
\end{align}
is proper relative to $\mathcal P$, for any non-negative function $w$.
\end{proposition}
The assumption of non-negative $S$ ensures the existence of the integral (which may be infinite, cf \eqref{eq:score}).  Together, Proposition \ref{unbiased} and \ref{integral} readily allow the construction of proper scoring rules based on summary statistics, as long as their estimators are unbiased. 
\begin{example}[$\mathcal F$-function]
The $\mathcal F$- or empty-space-function is defined for stationary point processes as the distribution function of the distance from the origin to the nearest point in a sampled point pattern.
It has the following unbiased estimator:
\[\wh {\mathcal F}(\bb y,r) := \sum_{x\in I_r}\frac{\mathbb{ 1}\{d(x,\bb y)\leq r\}}{\# I_r},\]
where $I$ is any finite regular grid of points, $I_r := I\cap W_{\circ r}$, and $W_{\circ r} = \{w\in W\,:\, b(w,r)\subset W\}$, see \citet[section 4.3]{MoellerWaagepetersen2003}. 
By Propositions \ref{unbiased} and \ref{integral},  the scoring rule
 \[S_{\mathcal F}(\bb y, F) := \int_\R(\wh {\mathcal F}(\bb y, r) - {\mathcal F}_F(r))^2\, w(r)dr,\]
is proper relative to the class of all stationary point process models, for any non-negative function $w$. Here, ${\mathcal F}_F$ denotes the (theoretical) $\mathcal F$-function under the point process model $F$.
\end{example}

\subsection{Scoring based on summary statistic estimators}

The construction of scoring rules above is quite intuitive, as it compares the estimator $\wh T$ to the true theoretical value $T_F$ under the predictive distribution. However, it is limited by two major restrictions. Firstly, for many summary statistics (such as, for example, the $K$-function and the intensity function) there are no unbiased estimators, see \citet[Chapter 4]{MoellerWaagepetersen2003}. Secondly, even if unbiased estimators exist, closed form expressions for the theoretical value $T_F$ may not be available for all point process models $F$. Both restrictions can be overcome by replacing $T_F$ by $\wh T(F)$, the pushforward probability measure of $F$ under the estimator $\wh T$. The following proposition is a direct consequence of the change-of-variables formula.
\begin{proposition}\label{biased}
Let $r\in \mathcal R$ be fixed. Denote by $\wh T(F,r)$ the pushforward distribution of $F$ under $\wh T(\cdot,r):W^\cup\to \R$. Consider a scoring rule $S$ on $\R$ that is proper relative to $\wh T(\mathcal P) := \{\wh T(F,r)\, ,\, F\in\mathcal P,r\in\mathcal R\}$. Then, the scoring rule
\[S_{\wh T}(\bb y , F) := S(\wh T(\bb y,r), \wh T(F,r))\]
is proper relative to $\mathcal P$.
\end{proposition}
Note that $S_{\wh T}$ is usually not strictly proper, even if $S$ is, since we might have $\wh T(F,r)=\wh T(G,r)$ for distributions $F\neq G$.  The key to rendering this result useful is the choice of the proper scoring rule $S$ on the real line. Note that we recover Proposition \ref{unbiased} if $\wh T$ is unbiased and we choose $S$ to be the SE in \eqref{eq:MSE}, $S(y,F) = (y-\E_F[X])^2.$ However, a preferable choice is the CRPS in \eqref{eq:CRPS} as it is {\it strictly} proper with respect to all distributions with finite first moment \citep{GneitingRaftery2007}. Choosing the CRPS allows to calculate $S_{\wh T}$ by Monte-Carlo approximation, without requiring knowledge of the pushforward measure $\wh T(F).$ When applying Proposition \ref{biased} with the CRPS formula in \eqref{eq:CRPS}, we obtain by the change-of-variables formula, supressing $r$ for brevity,
\begin{align*}
S_{\wh T}(\bb y,F) &= \E_{\wh T(F)}[|\wh T(\bb y) - X|] - \frac 1 2 \E_{\wh T(F)}[|  X' -  X|]\\
& = \E_{F}[|\wh T(\bb y) - \wh T(\bb X)|] - \frac 1 2 \E_{F}[|\wh T(\bb X') - \wh T(\bb X)|].
\end{align*}
Here, $X,X'$ are independent random variables with distribution $\wh T(F)$, and $\bb X'$ and $\bb X$ are independent point processes with distribution $F$. The latter expression can be approximated by Monte-Carlo sampling from the point process distribution $F$.

Another, somewhat surprising, advantage of this approach is that considering $\wh T(F)$ rather than $T(F)$ can increase the sensitivity of the score, in the following sense. It frequently happens that competing models have identical (theoretical) summary statistic $T(F)$: For example, homogeneous models with the same expected number of points all share the same intensity, and all Poisson models have identical $K$-functions. However, the estimated summary statistics $\wh T(F)$ are generally differently distributed for different models. Since the {\it strictly} proper CRPS is applied to $\wh T(F)$, misspecified models can be detected even if their theoretical summary statistics $T(F)$ fits the observation perfectly. For example, we show in the next section that the score based on estimators of the $K$-function assigns the lowest mean score to the correct model when several Poisson models are compared.

The main result of this section is the following corollary, applying Propositions \ref{integral} and \ref{biased} to the CRPS (which is non-negative).
\begin{corollary}[Summary statistic score]\label{sss}
Consider a non-negative weighting function $w:\mathcal R\to\R$ and an estimator for a summary statistic $\wh T$ that is integrable with respect to $F\otimes w(r)dr$ for all $F$ in $\mathcal P$. The scoring rule defined by 
\begin{align*}
S_{\wh T}(\bb y, F) &:= \E_{F}\bigg[\int_\mathcal R |\wh T(\bb y,r) - \wh T(\bb X,r)|w(r)\,dr\bigg]
 - \frac 1 2\E_{F}\bigg[\int_\mathcal R |\wh T(\bb X',r) - \wh T(\bb X,r)|w(r)\, dr\bigg]
\end{align*}
is proper relative to $\mathcal P$.
\end{corollary}

Note that both expectations can be calculated from random draws of the predictive point process distribution $F$, by Monte-Carlo approximation. An advantage of these scores are the very weak assumptions of this corollary which are satisfied for almost all point process models and summary statistic estimates. The weighting function $w(r)$ can be used to focus the attention of the score to certain areas, for example clustering at short range, when the $K$-function is considered, or the spatial distribution of points in a certain subregion of $W$, when the intensity is considered. For assessing the overall fit we will set $w = 1$ throughout the studies in this paper.

In this construction, $\wh T$ can be any real- or function-valued mapping satisfying the condition of Corollary \ref{sss}. In particular, no connection of the summary statistic estimator to the underlying summary statistic is required. The estimator does not need to be unbiased.  Below we list two examples which will be used to show the efficacy of our approach for point process model evaluation.

\begin{example}[Intensity score]\label{Ex:KernelEstimator}
The intensity function $\lambda$ of a point process is typically estimated by kernel estimators. These estimators are generally biased, making it impossible to apply Proposition \ref{unbiased}. 
For a kernel $k$ (i.e. a density on $W$) and a bandwidth $\sigma>0$, the kernel intensity estimator is based on the rescaled kernel $k_\sigma(w):= k(w/\sigma)$. It is defined as
\begin{equation}\label{eq:lambda hat}
  \wh \lambda_\sigma (\bb y,w) = \sum_{y\in \bb y} k_\sigma(w-y)/c_{W,\sigma}(y),
  \end{equation}
where $c_{W,\sigma}$ are edge correction factors defined as $c_{W,\sigma}(y) = \int_W k_\sigma(w-y)\,dw$. By Corollary \ref{sss} the intensity score, defined as
\begin{align}\label{IntScore}
S_{\wh\lambda_\sigma}(\bb y, F) := &  \E_F\bigg[\int_W|\wh \lambda_\sigma (\bb y,w) - \wh \lambda_\sigma (\bb X,w)|\,dw\bigg] -\frac 1 2 \E_F\bigg[\int_W|\wh \lambda_\sigma (\bb X',w) - \wh \lambda_\sigma (\bb X,w)| dw\bigg],
\end{align}
constitutes a proper scoring rule. Unless stated differently, we choose $k$ to be an isotropic Gaussian kernel with standard deviation 1, such that $k_\sigma$ has standard deviation $\sigma$. 
\end{example}
Since this score targets the intensity function, it assesses, roughly speaking, whether the predictive distribution has the correct spatial distribution and number of points, but neglects point interactions. 
On the other hand, if we are more interested in whether a predictive model reflects point interaction correctly, we can consider an estimator for the $K$-function.
\begin{example}[$K$-function score]\label{Ex:KFunScore}
The standard estimator for Ripley's $K$-function is defined as 
\begin{equation}\label{eq:K hat}
  \wh K(\bb y,r):=  \sum_{y_1\neq y_2\in\bb y} \frac{\mathbb{1}\{|y_1-y_2| < r\}}{\wh\lambda(y_1)\wh\lambda(y_2)|W\cap W_{y_1-y_2}|},
  \end{equation}
where $W_{y_1-y_2}$ denotes the shifted set $W+y_1-y_2$, and $\wh \lambda$ is a kernel estimator for the intensity. 
Thus, we obtain the proper $K$-function score
\begin{align}\label{Kfsc}
S_{\wh K}(\bb y, F) := & \int_0^{R} \E_F[|\wh K (\bb y,r) - \wh K (\bb X,r)|]\,dr -\frac 1 2 \int_0^R\E_F[|\wh K (\bb X',r) - \wh K (\bb X,r)|] \,dr,
\end{align}
where $R$ is an upper limit that should be chosen small relative to the diameter of $W$.
\end{example}
As $\wh K$ is sensitive to point interaction, this scoring rule specifically targets correct representation of point interaction in the predictive model. On the other hand, it will be relatively insensitive to misspecification of the intensity function, and for example, struggle to differentiate between different Poisson processes, which have the same $K$-function, see the example is Section~\ref{SS1}.

\subsection{Other scoring rules }

Depending on the set of forecast models that are compared, other proper scoring rules might be available: The logarithmic score can be evaluated when all predictions have known densities, in particular when all competing models are Poisson. This score is {\it strictly} proper and tests for significance of score differences tend to have higher power than for other scores, see \cite{Lerch&2017} and Section~\ref{SS2}. Therefore, the logarithmic score should be employed when available. Moreover, \cite{Brehmer&2021} introduce several new scoring rules for point processes. However, as the authors remark, most of these also require knowledge of the density or other mathematical properties of the process.  For example, they show that the scoring rule
\begin{align}\label{intensityscoreBrehmer}
S_\lambda(F,\bb y) := -\sum_{y_i \in \bb y} \log(\lambda_F(y_i)) + n(\bb y)|\Lambda_F| + c(|\Lambda_F| -n(\bb y))^2
\end{align}
is proper for any $c>0$, where $|\Lambda_F|= \int_W \lambda_F(w)dw$, and $n(\bb y)$ denotes the number of points in the point pattern $\bb y$. Since this score evaluates the fit of the intensity of the predictive model, it provides an alternative to the intensity estimator score in \eqref{Ex:KernelEstimator}. Evaluating this score requires knowledge of the intensity for all considered prediction models. Intensities are known for most double stochastic Poisson processes, also known as Cox processes, but only for a few models of the Markov type \citep{MoellerWaagepetersen2003}.

Another option to create scoring rules sensitive to the intensity is by dividing the observation space $W$ into a large number $N$ of spatial bins. This reduces observed point patterns $\bb y \in W^\cup$ to counts of points per bin, and  proper scoring rules can be applied to the counts, see \citet{Brehmer&2021} for a discussion. This method is frequently applied in the evaluation of earthquake rate predictions \citep{Schorlemmer&2018}. It can be viewed as an application of Proposition \ref{biased}, since the binning constitutes a mapping from the space of point patterns into a simpler space where scoring rules are available. 

\subsection{Computational aspects}

In practice, for evaluating the scores proposed in Corollary \ref{sss}, $n$ realizations $T_1,...,T_n$ of $\wh T(\bb X), \bb X \sim F$ need to be generated in order to approximate the score by 
\[\frac{1}{n}\sum_{i = 1}^n\bigg[\int_\mathcal R |\wh T(\bb y,r) - T_i(r)|w(r)\,dr\bigg]
 - \frac 1 {2n(n-1)}\sum_{i\neq j}\bigg[\int_\mathcal R |T_i(r) -  T_j(r)|w(r)\, dr\bigg],\]
 where the integrals are calculated numerically. 
 This generates a trade-off in that increasing $n$ reduces the uncertainty of the Monte-Carlo approximation but increases
computation costs for the score evaluation. As a reference, the simulation study presented in Section \ref{SS2} uses $n = 300$ and requires the computation of 9000 scores in total. These calculations take approximately 15 minutes, running on a laptop with 32 GB RAM, using a single core of a quadcore processor with 2.9 GHz. 

In many point process applications, the number of available observations is limited, diminishing the drawback from using scoring rules that are computationally costly. In long-range earthquake forecasting, for example, evaluation periods typically include five years or more of earthquakes, and new observations therefore take a long time to materialize. In ecology, collecting observations usually requires field studies and work-intensive data labeling, similarly resulting in a small number of available observations. In such cases, it may be particularly important to aim to minimize the uncertainty of the Monte-Carlo approximation.


\section{Simulation study 1: Detecting differences in intensity and interaction}\label{SS1}

In this simulation study, we analyze some properties of the proposed intensity- and the $K$-function score.

\subsection{Competing models}

We consider the five different point process models listed in Figure \ref{models}. For study area, we consider the spatial window $W = [0,10]\times [0,10].$ The first two models (hP and hP+) are homogeneous Poisson processes with 50 and 60 expected points, respectively. The third model (ihP) is an inhomogeneous Poisson process with 50 expected points, and with an intensity that increases linearly in the distance from the lower left corner of the window. The fourth model (Str) is a homogeneous Strauss process, in which the points repel each other, and the typical point pattern is more regular than for the homogeneous Poisson process. The Strauss process is defined by its density
\[f(\bb x) = c\beta^{n(\bb x)}\gamma^{s_R(\bb x)},\]
where $c$ is an untractable normalizing constant, $\beta>0, R>0$, and $\gamma\in (0,1)$ are parameters, $n(\bb x)$ denotes the number of points in $\bb x$ and $s_R(\bb x)$ is the number of pairs of points in $\bb x$ with distance less than $R$. The value $R$ is the range of interaction between points, and $\gamma$ determines the strength of the interaction, with smaller $\gamma$ leading to  stronger inhibition between close points. We set $\gamma = 0.5$, $R = 1$ and $\beta = 1.15$, resulting in an expected number of points of approximately 50, the same as for the models hP and ihP. 
\begin{figure}[h]
  \centering
\includegraphics[width = 0.19\textwidth]{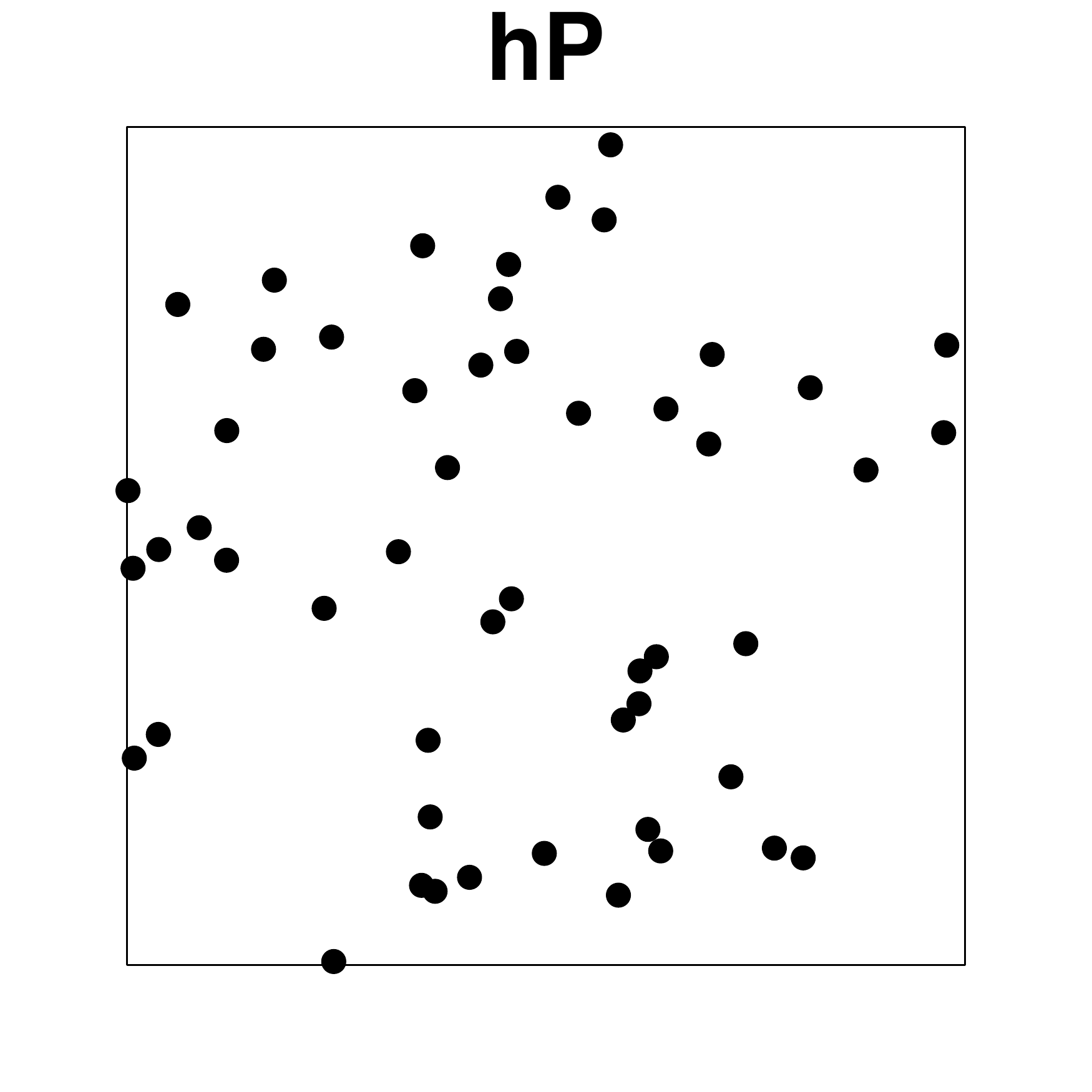}
\includegraphics[width = 0.19\textwidth]{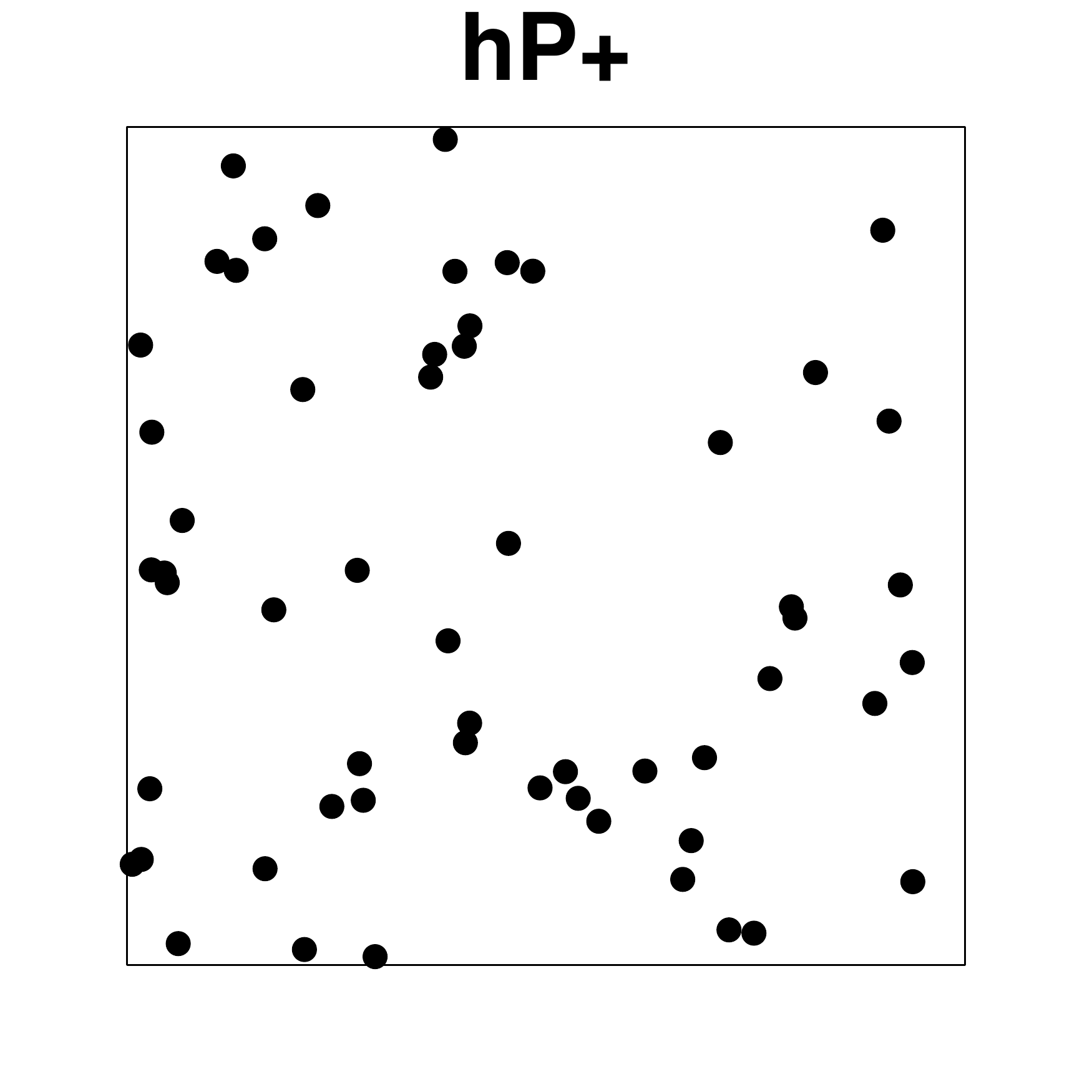}
\includegraphics[width = 0.19\textwidth]{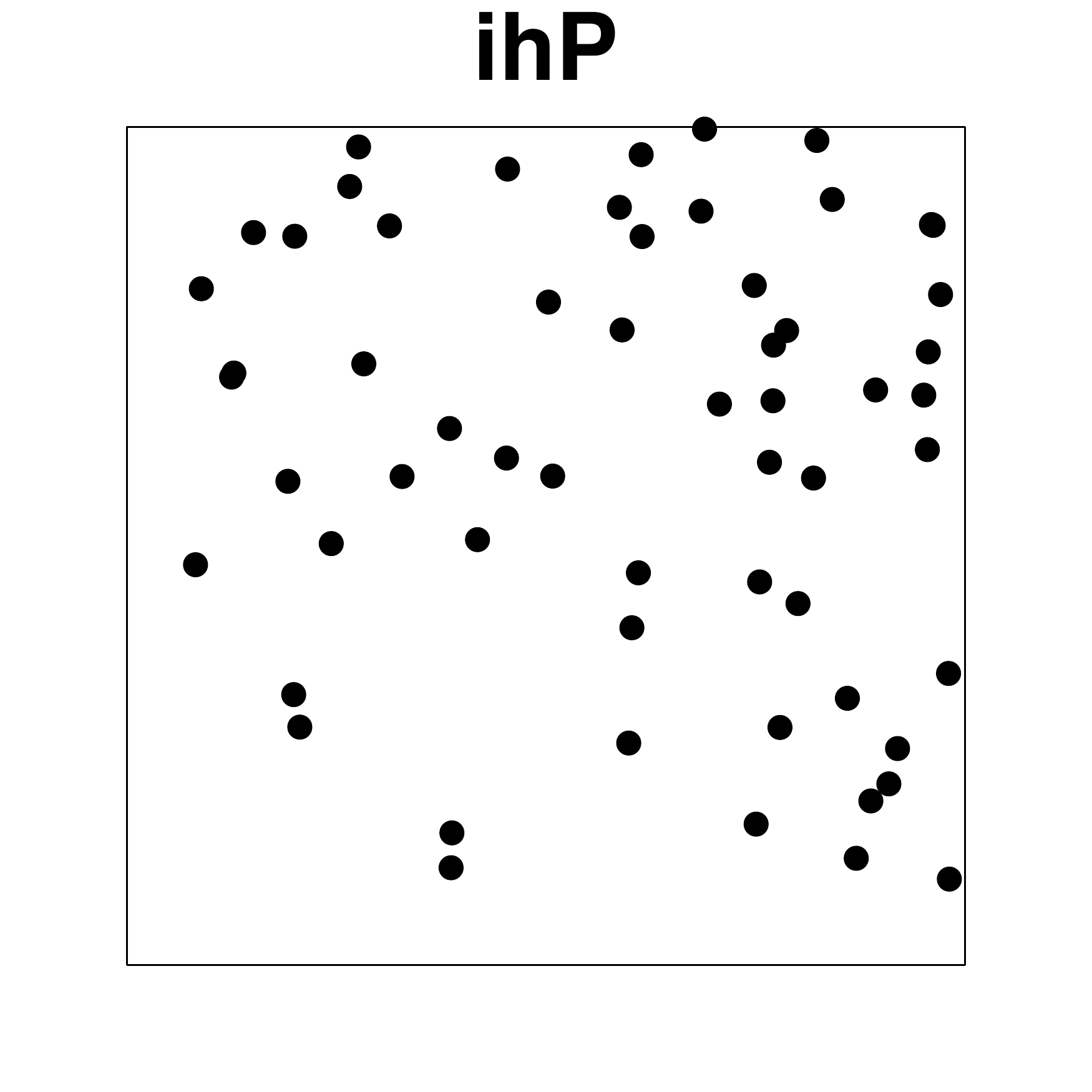}
\includegraphics[width = 0.19\textwidth]{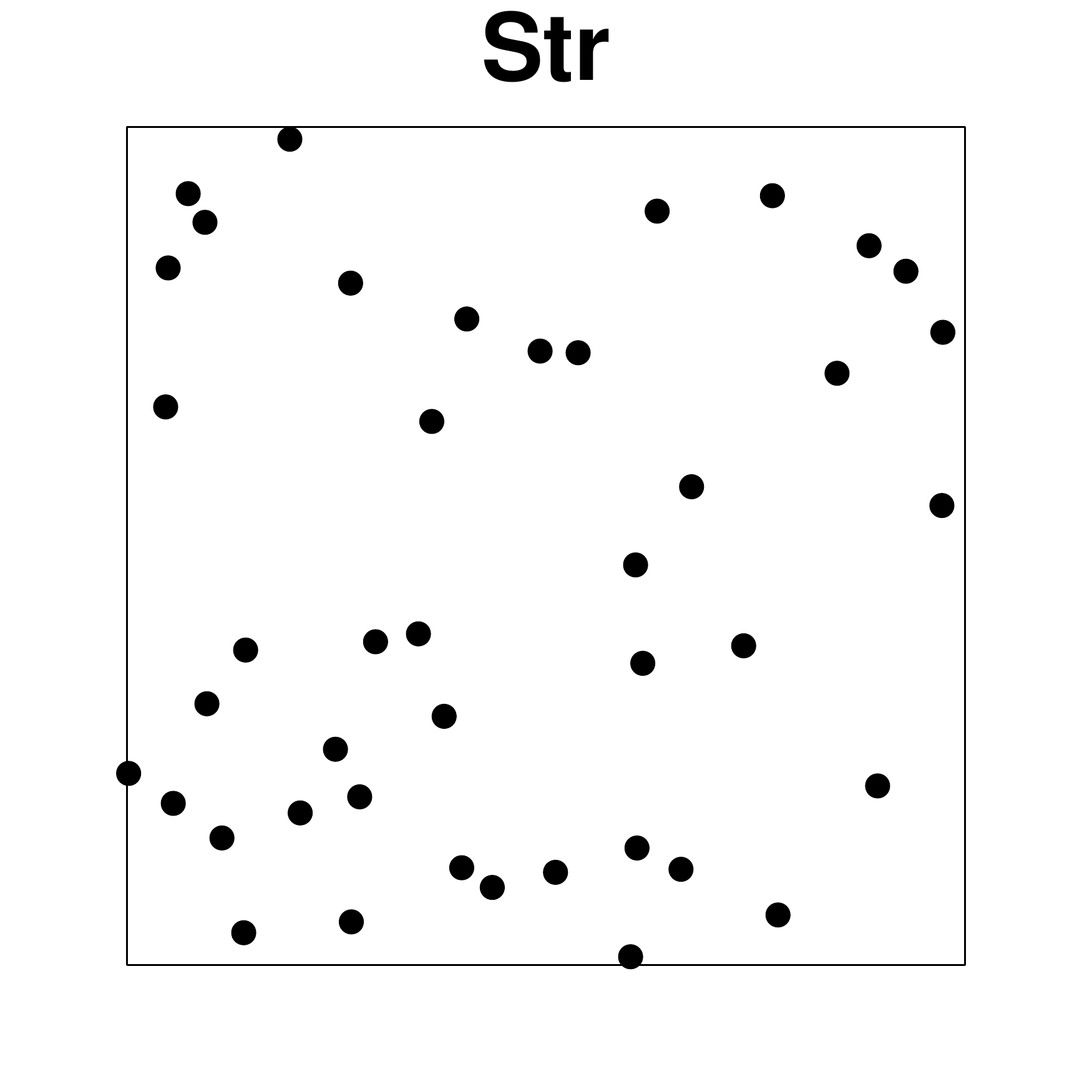}
\includegraphics[width = 0.19\textwidth]{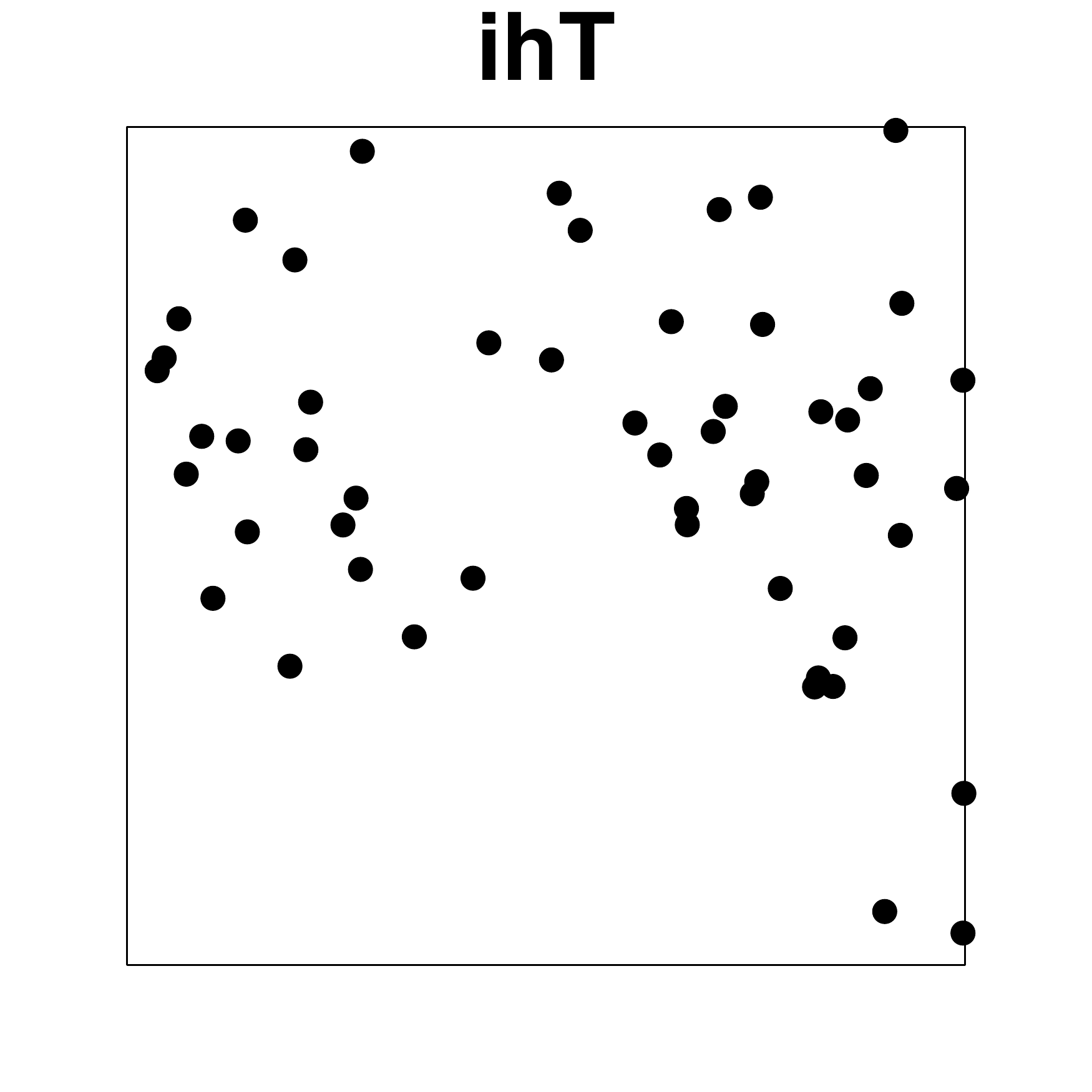}\\[1em]
\begin{tabular}{l|lllr}
Name	&			Model					& Intensity					& Point interaction 		& $\E[n(\bb X)]$	\\\hline
hP 	& 			homogeneous Poisson		& $ 1/2$			 		& none 						& 50				\\
hP+ &		homogeneous Poisson			& $ 3/5 $		& none						& 60				\\
ihP	&		inhomogeneous Poisson		& $\sim \sqrt{x^2+y^2}\ $	& none						& 50				\\
Str	&		homogeneous Strauss 		& $\approx 1/2$					& inhibition				& $\approx 50$		\\
ihT	& 		inhomogeneous Thomas 		& $\sim \sqrt{x^2+y^2}$		& clustering				& $\approx 50$
\end{tabular}\\[0.5em]
\caption{ Example plots (top row) and characteristics of the five models considered in the first simulation study. The spatial window is $[0,10]\times [0,10]$, i.e. for the inhomogeneous processes, the intensity increases with distance from the lower left corner.\label{models}}
\end{figure}
The fifth model is an inhomogeneous Thomas process (ihT), which is constructed by generating an (invisible) Poisson process of parent points, where each parent then generates a random number of offsprings that are spatially distributed according to a Gaussian kernel centered at the parent point. Here, we choose an inhomogeneous parent process with intensity equal to half of that of model ihP, the number of offsprings per parent is Poisson distributed with mean 2, and the standard deviation of the (homogeneous) Gaussian dispersion kernel is set to 0.5. By these choices, the Thomas process has an intensity similar to model ihP, and the number of expected points within the observation window is again approximately 50. As seen in the examples in Figure \ref{models}, simulations from these models may be difficult to tell apart by eye and the models are chosen to challenge the scoring functions.

\subsection{Evaluation}

We consider each of the five models both as true distribution $G$ and as predictive distribution $F$, for a total of 25 combinations. For each combination we estimate $\E_G[S_{\wh T}(\bb Y , F)]$ by simulating 100 independent copies of $\bb Y\sim G$ and averaging $S_{\wh T}(\bb Y , F).$ For the computation of $S_{\wh T}(\bb Y , F)$ the expectations are approximated by Monte-Carlo approximation based on 100 independent draws from $F$. As summary statistic estimator $\wh T$ we consider both the kernel estimator $\wh \lambda$ in \eqref{eq:lambda hat} and the K-function estimator $\wh K$ in \eqref{eq:K hat}. All computations in this paper are carried out using the \texttt{R}-package \texttt{spatstat} \citep{spatstat}. For the kernel estimator we use bandwidth $\sigma$=1.25 which is the default value (for the considered spatial window) of the \texttt{spatstat} functions \texttt{density.ppp} and \texttt{Kinhom.ppp} that are employed to calculate the estimators. 
\begin{figure}
\centering
\begin{subfigure}{0.475\textwidth}
\includegraphics[width = \textwidth]{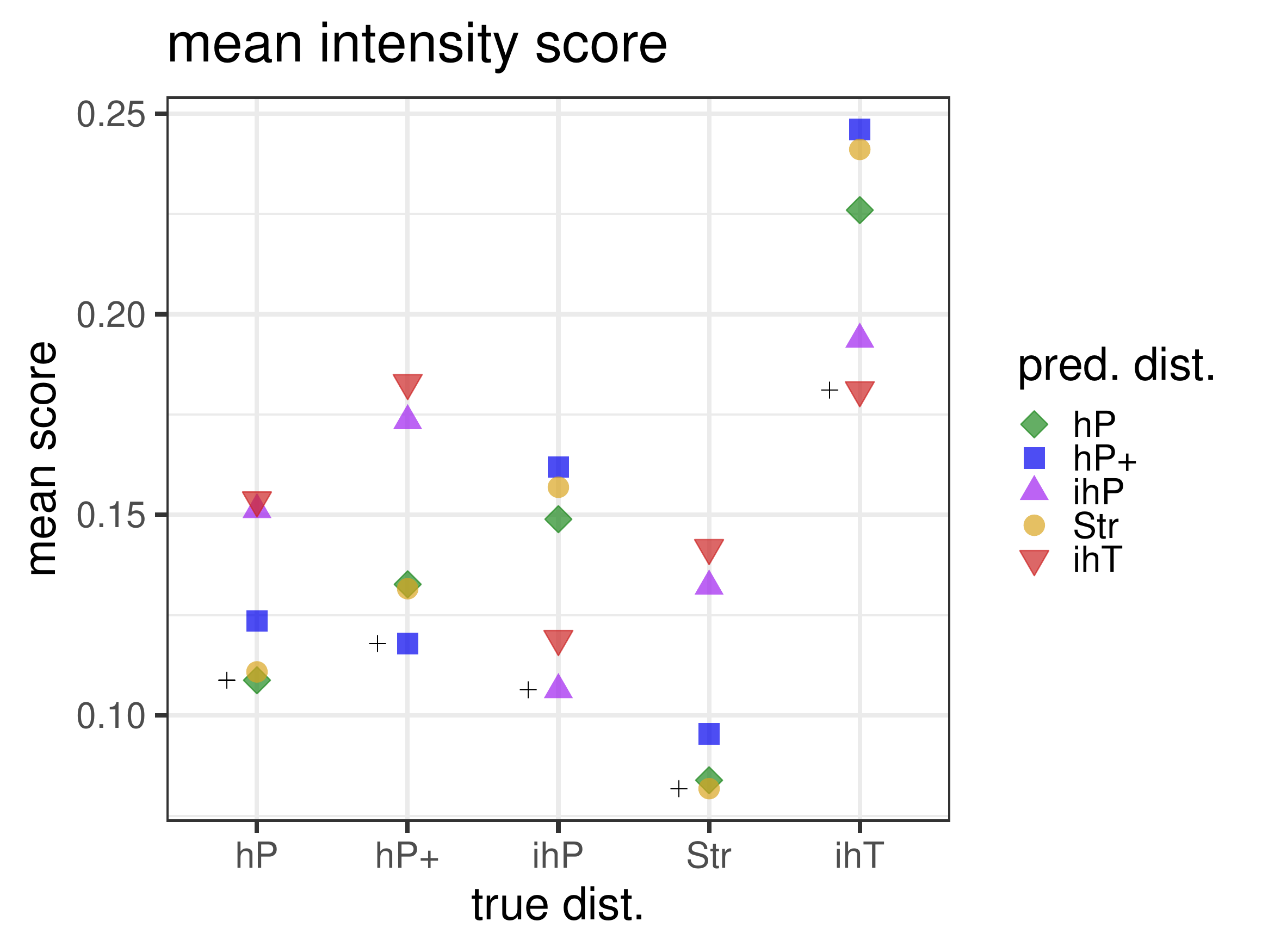}
\end{subfigure}\hspace{1em}
\begin{subfigure}{0.475\textwidth}
\includegraphics[width = \textwidth]{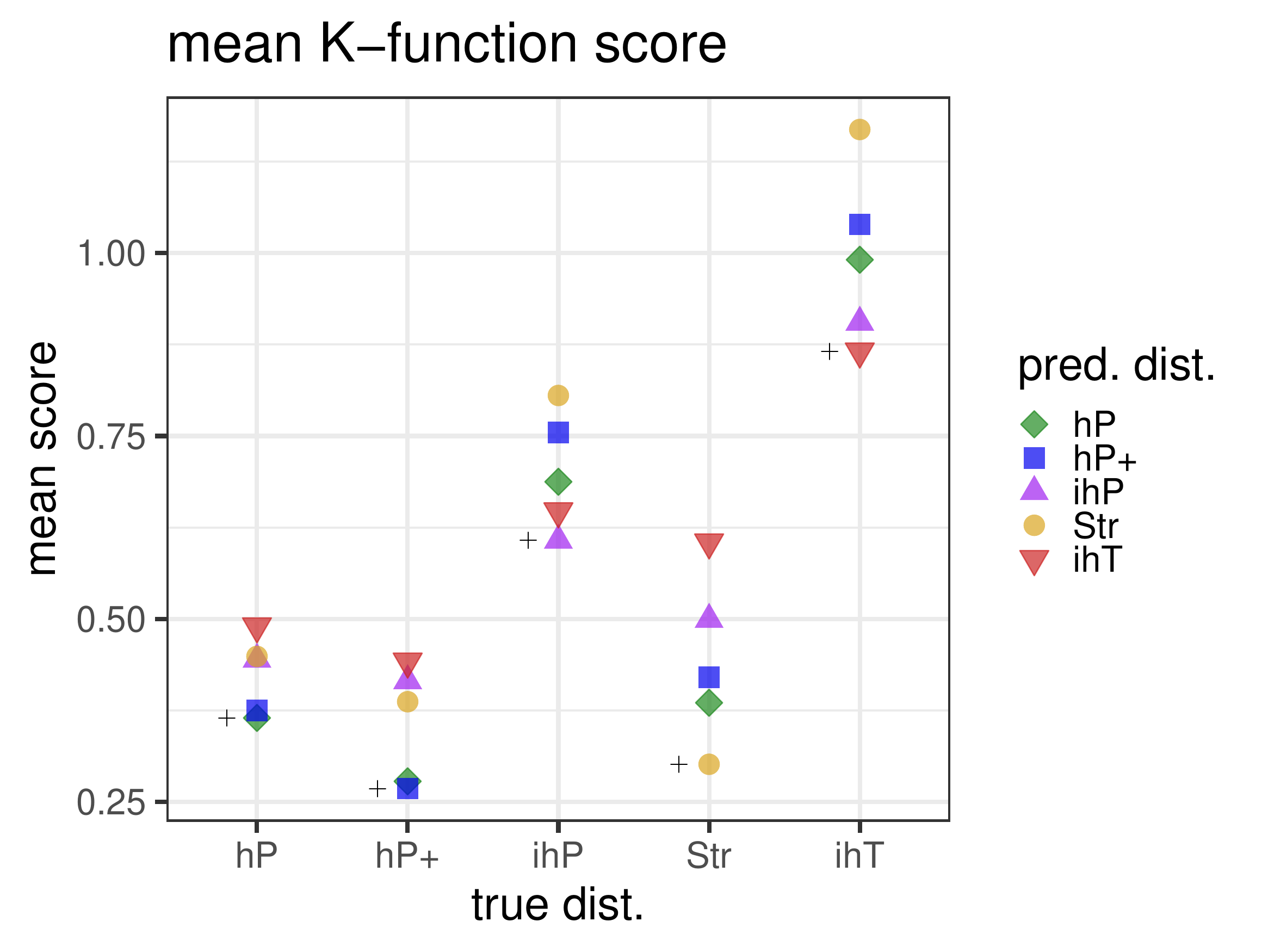}
\end{subfigure}
\caption{The mean scores $\E_G[S_{\wh \lambda}(\bb Y,F)]$ (left hand side) and $\E_G[S_{\wh K}(\bb Y,F)]$ (right hand side) for each combination of the five competing models. The $x$-axis shows the true distribution $G$ of the data, and the score for the correct distribution is additionally marked by a plus. \label{scores}}
\end{figure}

The results are presented in Figure \ref{scores}. For both scoring rules, the expected score is minimized for the true model under all distributions. The results highlight the sensitivity of the score to the underlying summary statistic. For example, under the intensity score, the mean scores for the models hP and Str are very similar, since these models have nearly identical intensities. On the other hand, the intensity score very clearly separates the homogeneous from the inhomogeneous models. Unlike for the intensity score, the $K$-function score identifies a clear difference between the models hP and Str. However, the differences in the $K$-function scores are particularly small when selecting between the two homogeneous Poisson models hP and hP+. 

\subsection{Significance of the results}
The significance of the mean score differences in Figure \ref{scores} can be assessed by permutation tests as e.g. described in \citet{Moeller&2013} and \citet{Good2013}. Table \ref{permtest} shows the $p$-values of permutation tests for all combinations of observed and predicted model. The results underline that the scores are targeting different properties: While the score differences for $S_{\wh \lambda}$ are always significant at a $5\%$ level, when observed and predicted model have different intensity, the score differences under $S_{\wh K}$ are significant whenever they have different point interaction. In practice, it is therefore advisable to apply scoring rules sensitive to properties that are most important for the study at hand, and, in case of doubt, compare competing models based on a variety of attributes. 

Note that the differences between the mean $K$-function scores of the different Poisson models are not significant at the 5\% level. As these models have identical theoretical $K$-function, this is to be expected. It is therefore rather remarkable that the mean score gets minimized in the correct model for all model combinations.
This follows from the estimated $K$-functions $\wh K$ having different distributions under the different Poisson models, and these differences being identified by the strictly proper CRPS. A similar effect is e.g. reported in \citet{Thorarinsdottir&2016} in the context of assessing the multivariate structure of high-dimensional real-valued predictions. 

\setlength{\tabcolsep}{2pt}

\begin{table}
\footnotesize
\centering
\begin{tabular}{llccccc}
\multicolumn{7}{c}{$S_{\wh \lambda }(Y,F)$}\\[0.5 em]\hline\hline
\multicolumn{2}{l}{$F:$}& hP & hP+ & ihP & Str &ihT\\ \hline
\multirow{5}{*}{$G$:\,} 
& \multicolumn{1}{l|}{hP}  	& - 		& $<$0.1 	& $<$0.1 			& {\bf \phantom{0}8.0} 	& $<$0.1\\
& \multicolumn{1}{l|}{hP+}  & $<$0.1 	& - 		& $<$0.1 			& $<$0.1 				& $<$0.1 \\
& \multicolumn{1}{l|}{ihP}  &  $<$0.1 	& $<$0.1 	& - 				& $<$0.1 				& {\bf \phantom{0}8.4}\\
& \multicolumn{1}{l|}{Str } & {\bf 15.1}& $<$0.1 	& $<$0.1 			& - 					& $<$0.1 \\
& \multicolumn{1}{l|}{ihT } &  $<$0.1 	& $<$0.1 	& \phantom{0}4.4 	& $<$0.1 				& - 
\end{tabular}\hspace{6em}
\begin{tabular}{llccccc}
\multicolumn{7}{c}{$S_{\wh K }(Y,F)$}\\[0.5 em]\hline\hline
\multicolumn{2}{l}{$F:$}& hP & hP+ & ihP & Str &ihT\\ \hline
\multirow{5}{*}{$G$:\,} 
& \multicolumn{1}{l|}{hP}  	& - 					& {\bf 38.3} 			& {\bf 17.6} 		& $<$0.1 	& $<$0.1 \\
& \multicolumn{1}{l|}{hP+}  & {\bf 38.9} 			& - 					&  {\bf 16.7} 		& $<$0.1 	& $<$0.1 \\
& \multicolumn{1}{l|}{ihP}  & {\bf \phantom{0}8.2}  & {\bf \phantom{0}6.2}	& - 				& $<$0.1 	& {\phantom{0}0.5}\\
& \multicolumn{1}{l|}{Str } & $<$0.1				& $<$0.1 				& $<$0.1 			& - 		& $<$0.1 \\
& \multicolumn{1}{l|}{ihT } &  $<$0.1				& $<$0.1  				& {\phantom{0}0.7} 	& $<$0.1 	& - 
\end{tabular}
\caption{$p$-values (in $\%$) of permutation tests assessing the significance of the difference between the score of predictive distribution $F$ and the score of the true distribution $G$.
Values above $5$\% (in bold) indicate nonsignificance, and the score cannot reliably distinguish between $F$ and $G$ in the setting of simulation study 1 and a sample size of 100 point patterns.\label{permtest} }
\end{table}

\section{Simulation study 2: The role of the bandwidth for the intensity score}\label{SS2}

In a second simulation study, we analyze the effect of the bandwidth parameter $\sigma$ in the intensity score, and compare its performance to the logarithmic score. Specifically, we assess how reliably the logarithmic score and the intensity score  identify the correct model when the mean score is calculated from samples (or observations) of varying sizes. This study uses an observation window of the same size as before but with approximately twice as many expected points. We consider a range of bandwidth parameters $\sigma$ to analyze how the bandwidth affects the ability of the score to identify the correct model.

\subsection{Competing models}

Here, we focus on models where the density is available in closed form for a straight-forward evaluation of the logarithmic score. As observation-generating model we consider a Poisson model with a standard Gaussian kernel (multiplied by a scalar) as intensity function. The competing predictive models are various Poisson processes based on Gaussian kernels, which allows us to look at the effect of a variety of different misspecifications in the prediction: Local shift (wrong mean $\mu$ of the kernel), misspecified spread (wrong standard deviation $\eta$), misspecified shape (wrong correlation $\rho$), as well as misspecification in the expected number of points, see Figure \ref{modelsPois}. We denote the models by $F_1,...,F_6$, with $F_1$ being the observation model.  Again, the model parameters are chosen such that it is challenging to tell the models apart by eye. 

\begin{figure}[!hbpt]
\hspace{-2em}
\begin{minipage}{0.23\textwidth}
\includegraphics[width = \textwidth]{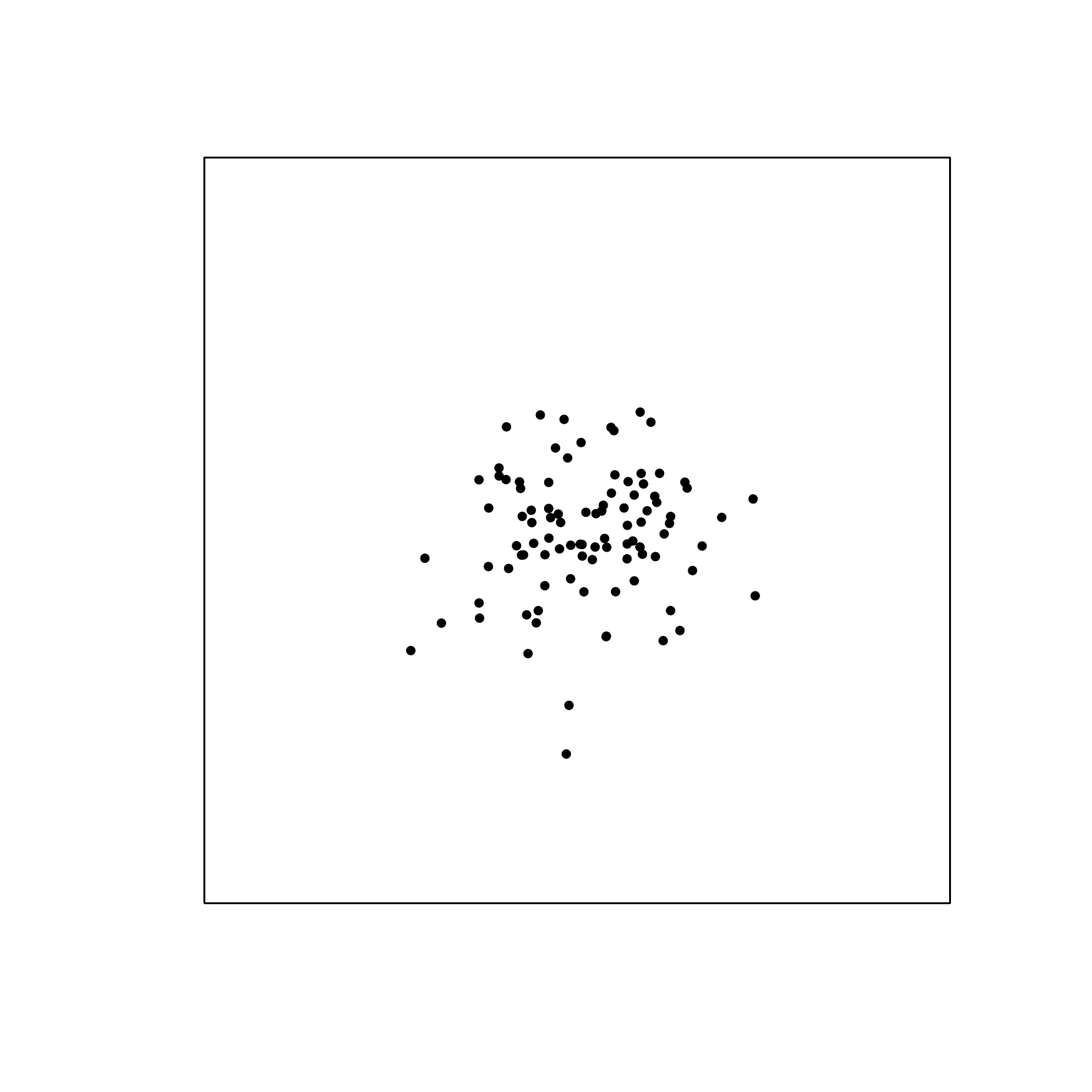}\\[-2em]
\centering
$\phantom{(}\text{obs. model}\phantom{)}$
\end{minipage}\hspace{-3.9em}
\begin{minipage}{0.23\textwidth}
\includegraphics[width = \textwidth]{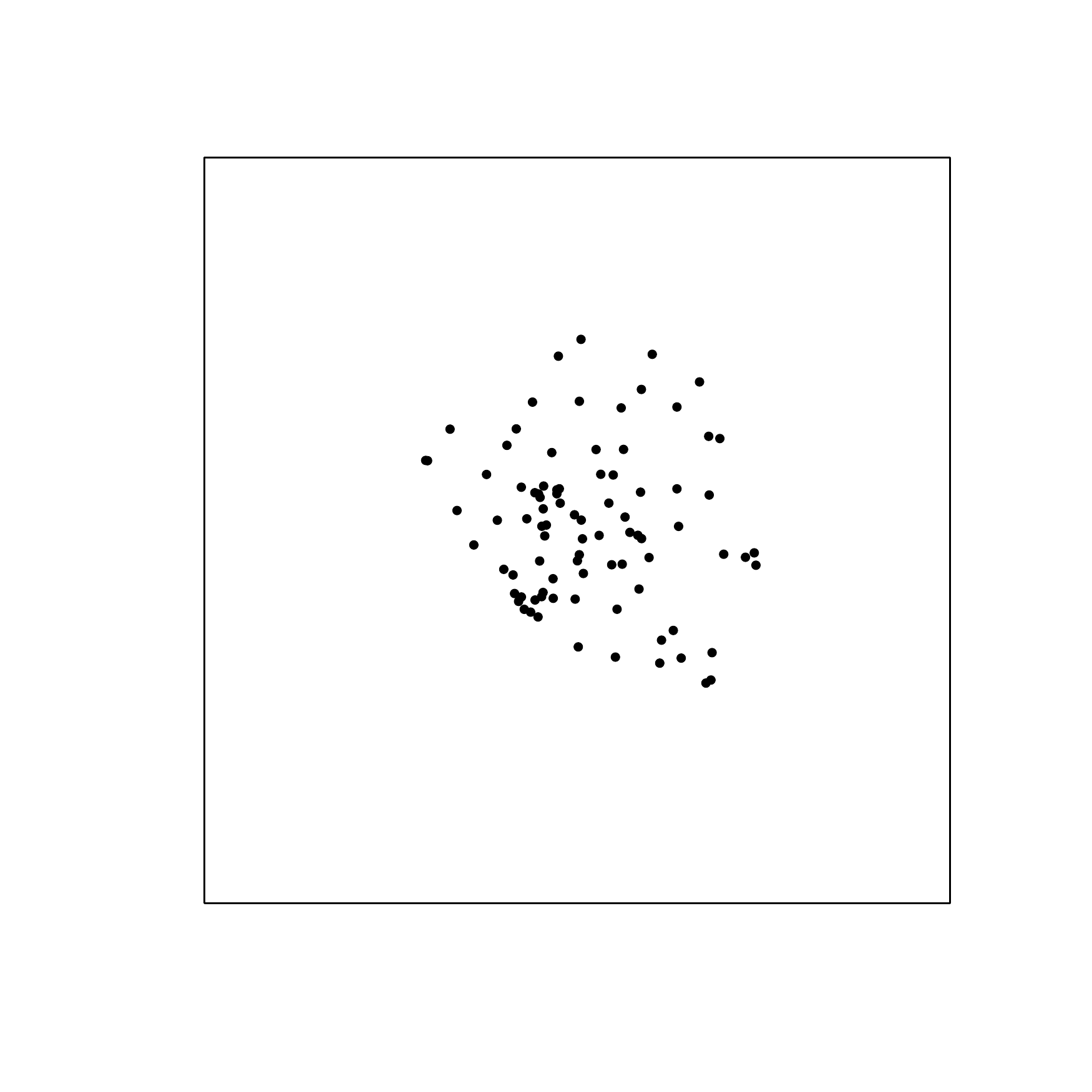}\\[-2em]
\centering
$\mu = (0.1,0)$
\end{minipage}\hspace{-3.9em}
\begin{minipage}{0.23\textwidth}
\includegraphics[width = \textwidth]{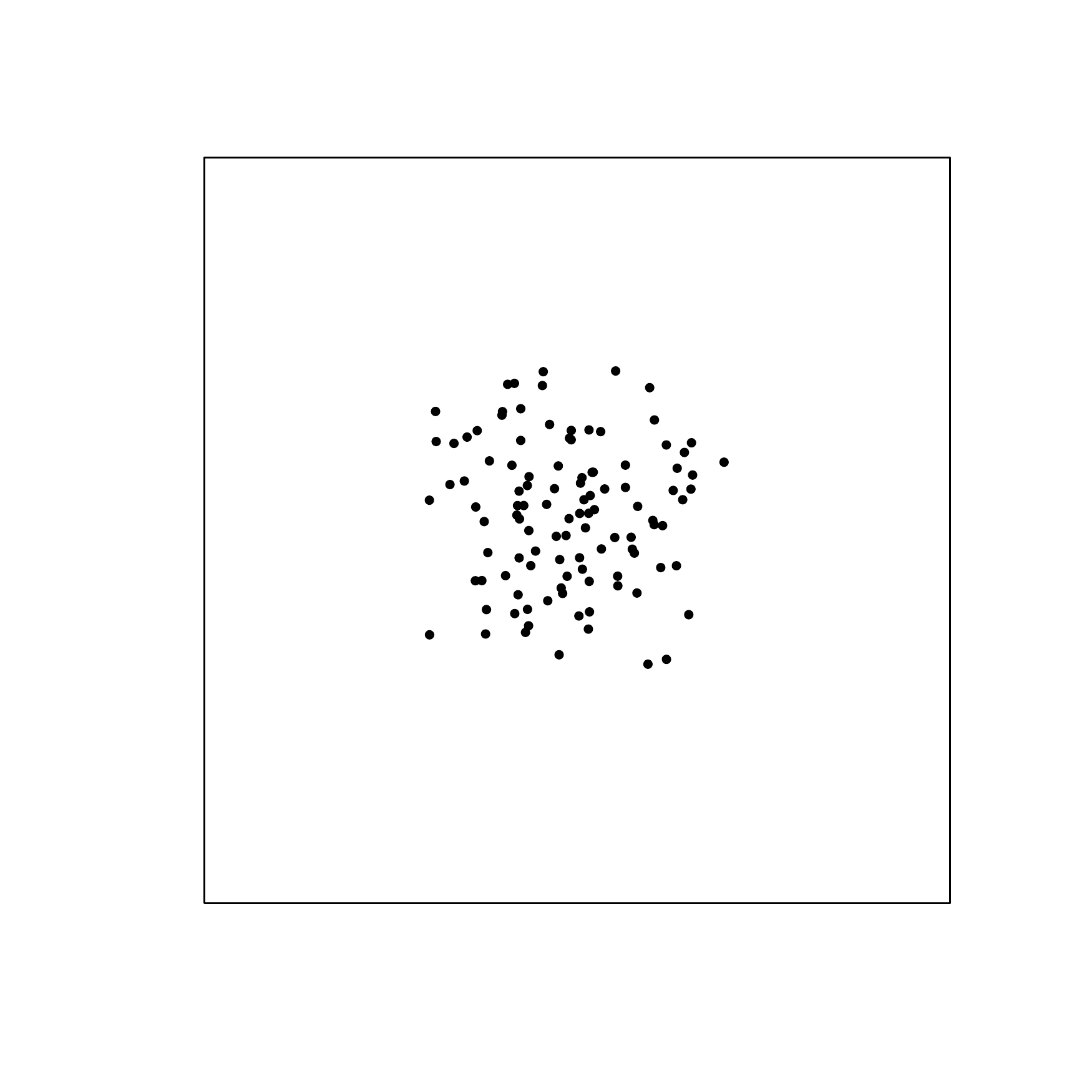}\\[-2em]
\centering
$\eta = (0.9,0.9)$
\end{minipage}\hspace{-3.9em}
\begin{minipage}{0.23\textwidth}
\includegraphics[width = \textwidth]{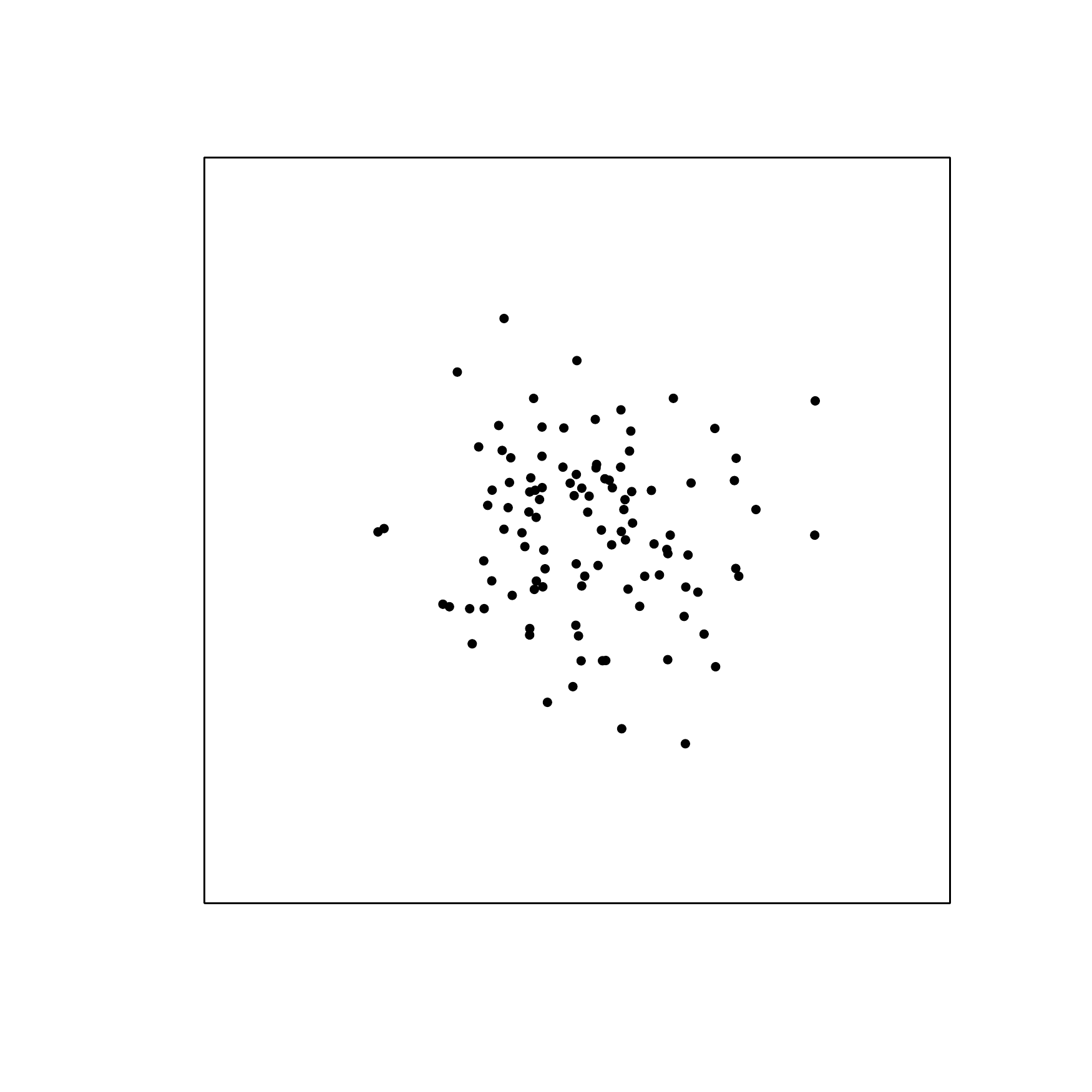}\\[-2em]
\centering
$\eta = (1.1,1.1)$
\end{minipage}\hspace{-3.9em}
\begin{minipage}{0.23\textwidth}
\includegraphics[width = \textwidth]{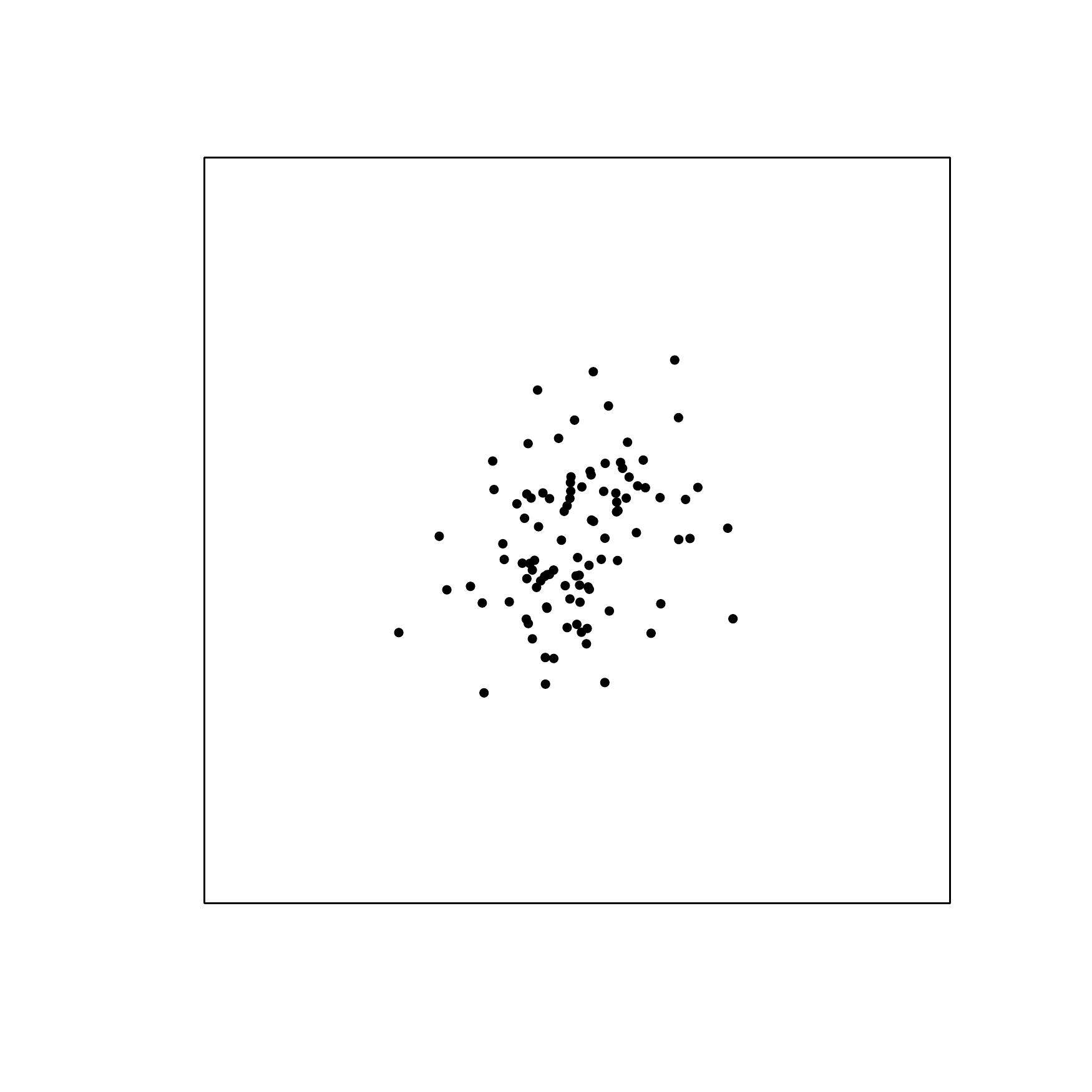}\\[-2em]
\centering
$\rho = 0.1$ 
\end{minipage}\hspace{-3.9em}
\begin{minipage}{0.23\textwidth}
\includegraphics[width = \textwidth]{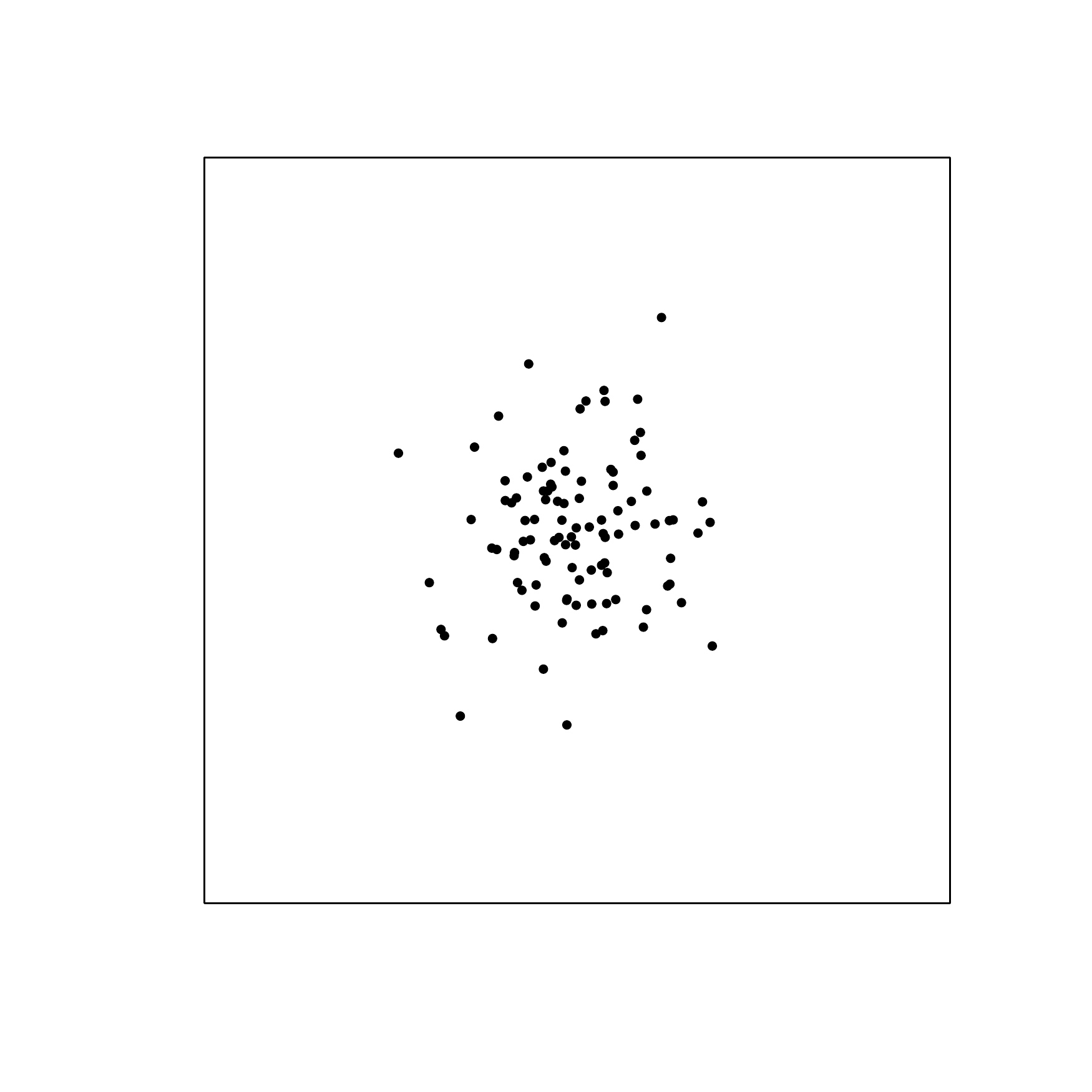}\\[-2em]
\centering
$\E[n(\bb X)] = 105$
\end{minipage}
\caption{ 
The different Poisson models considered in the second simulation study. The spatial window is $ W = [-5,5]\times[-5,5]$.  From left to right: The observations are distributed according to $F_1$,  a Poisson model with the intensity equal a standard Gaussian kernel, multiplied by 100. Model $F_2$ has an intensity with misspecified mean $\mu$. Models $F_3$ and $F_4$ have mean zero and a misspecified standard deviation $\eta = 0.9$ and $\eta = 1.1,$ respectively (both in $x$- and $y$-component). In model $F_5$, a mean-zero unit-variance Gaussian kernel with correlation of 0.1 between $x$- and $y$-component is used. The intensity of model $F_6$ is again standard Gaussian, but multiplied by 105, leading to a slightly larger number of expected points.
 \label{modelsPois}}
\end{figure}

\subsection{Evaluation}
We generate 300 random samples $\{\bb Y_j\}_{j = 1,...,300}$ from the observation model $F_1$. Then, for each of the predictive models $F_1,...,F_6$ we compute the 300 score evaluations $S(\bb Y_j,F_i)$, for the logarithmic score and the intensity score with varying bandwidth parameters $\sigma$. The latter requires Monte-Carlo approximation of the two expectations in \eqref{IntScore}, which are each based on 300 i.i.d. realizations from the predictive models. We then subsample from the collection of scores $\{S(\bb Y_j,F_i)\}_{j = 1,...,300;\, i = 1,...,6}$ in order to assess the performance of the mean score under differently sized observation sets.

\begin{figure}[t]
\centering
\begin{subfigure}[t]{0.47\textwidth}
\includegraphics[scale = 0.3]{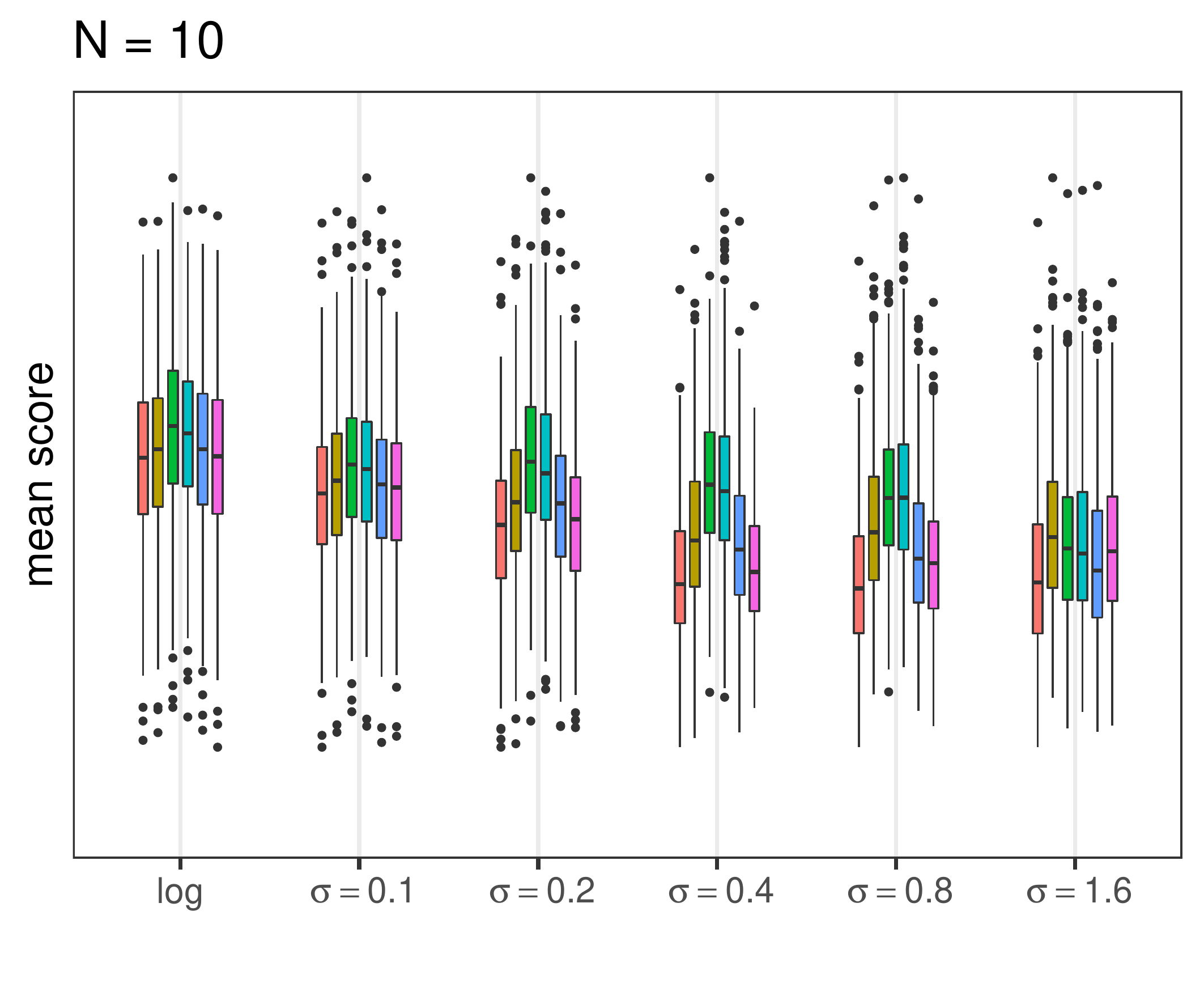}
\end{subfigure}\hspace{1em}
\begin{subfigure}[t]{0.47\textwidth}
\includegraphics[scale = 0.3]{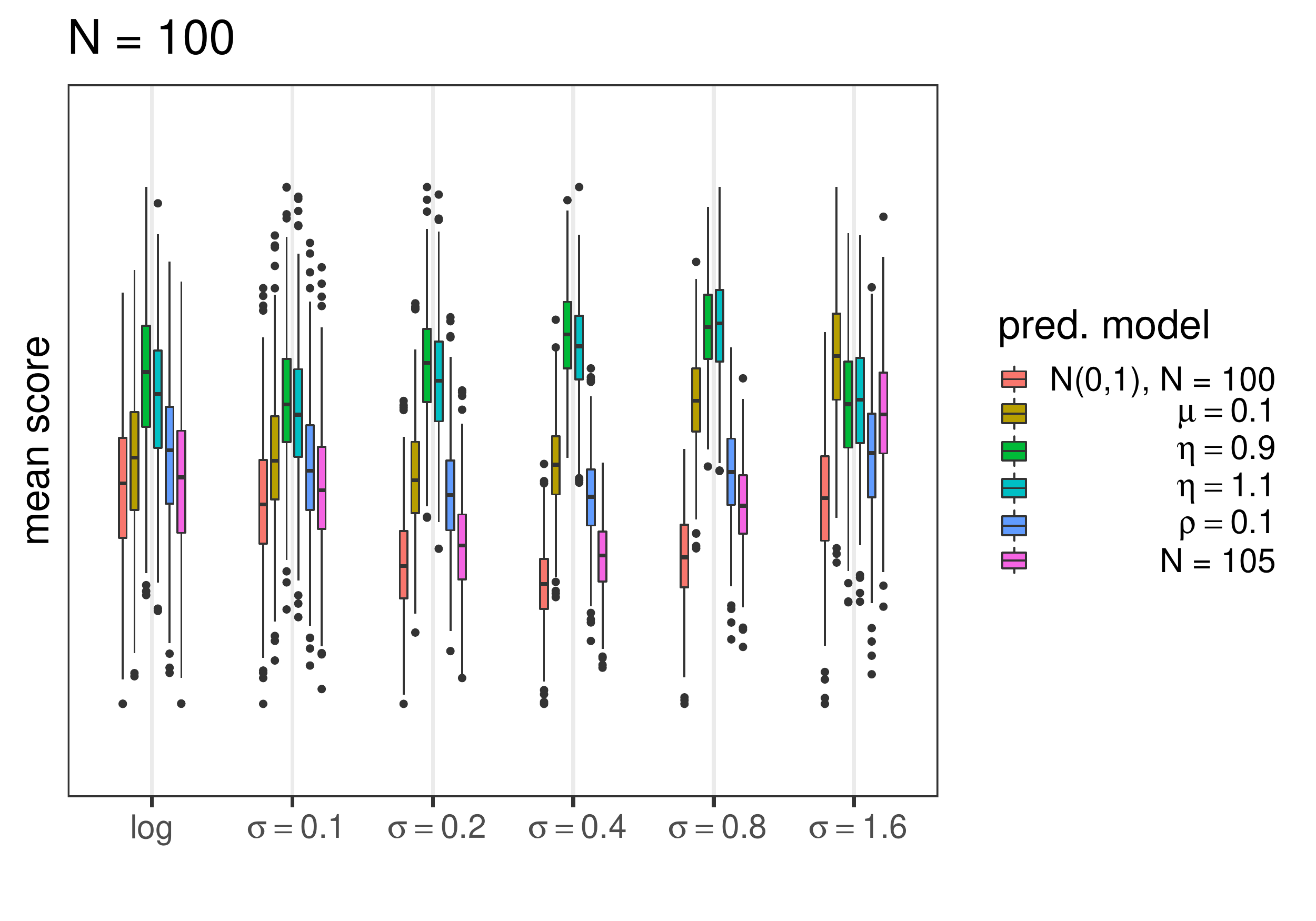}
\end{subfigure}
\caption{Mean scores under true distribution $F_1$. The logarithmic score is compared with the intensity score for a range of different bandwidths $\sigma$. The boxes indicate variability in the mean scores based on 10 (left) or 100 (right) observations. \label{scorePoisson}}
\end{figure}

Figure \ref{scorePoisson} shows sample distributions (as boxplots) of the mean scores when 10 (resp. 100) observations are available. For ease of comparison, we apply affine positive linear functions to all scores to fit them to the same scale which does not affect the ordering of competing models. When the number of available observations is $N=10$, the uncertainty in the estimated mean scores is much larger than the score differences between different predictive models. For $N=100$, the differences in mean score increase relative to the spread, indicating that the scores can clearer identify the correct predictive distribution. The figure indicates quite good performance of the intensity score with bandwidths $\sigma=0.2, \, 0.4,\,0.8$. 

\subsection{Significance of the results}

As competing models should be evaluated on the same set of observations \citep{GneitingRaftery2007}, there is potential correlation in the score values for competing models. Figure \ref{fig:correlation} shows how the correlation in score values for models $F_1$ and $F_5$ varies between the scores: the correlation is stronger for the logarithmic score than the intensity score, and, for the latter, the correlation changes with the bandwith. A stronger correlation can counteract small differences between marginal distributions of score values, which renders the results in Figure \ref{scorePoisson} inconclusive regarding which score is most likely to prefer the correct model. To account for the joint distribution of score values, Figure \ref{pvalsPoisson} shows (approximated) $p$-values for permutation tests comparing the correct model to each of the competing models. 

\begin{figure}[t]
\centering
\includegraphics[scale = 0.2]{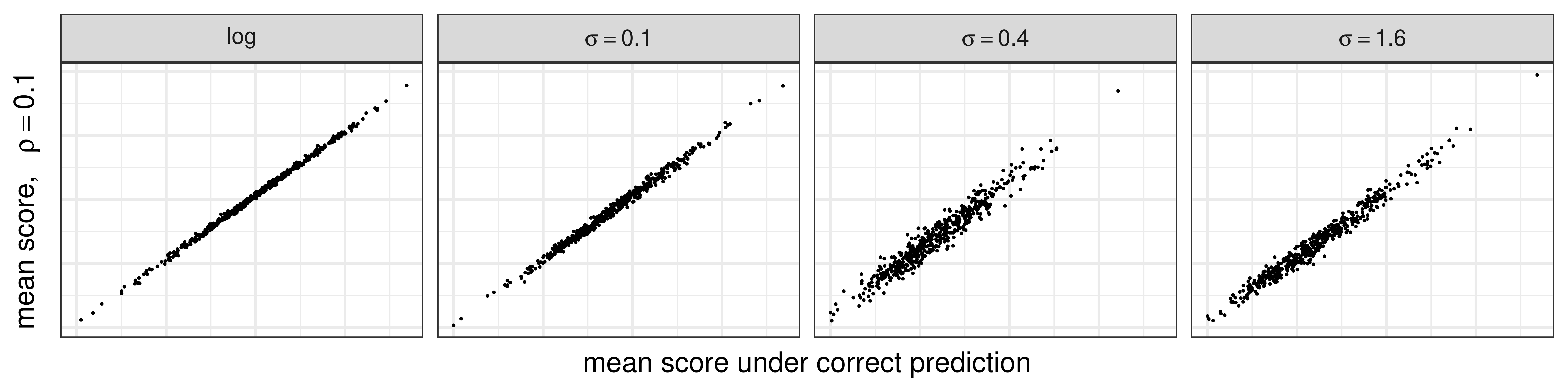}
\caption{Scatter plot of mean scores for the two models $F_1$ and $F_5$, evaluated on 10 random observations distributed according to $F_1$. For each dot the coordinates correspond to the two mean scores evaluated on the same (synthetic) observation.\label{fig:correlation}}
\end{figure}

\begin{figure}[t]
\centering
\includegraphics[width = \textwidth]{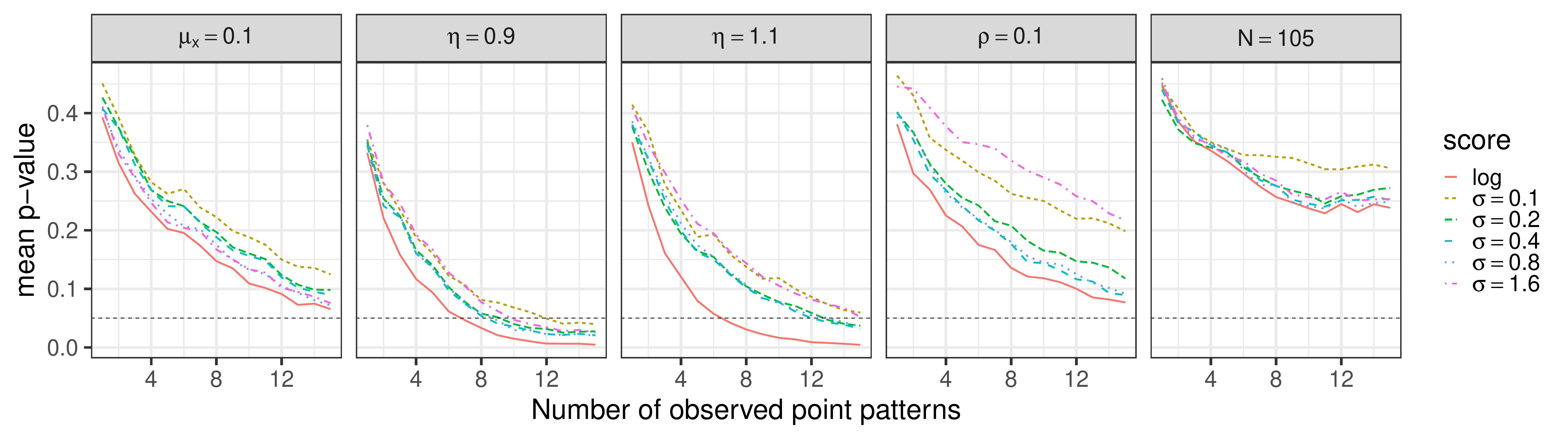}
\caption{ Estimated $p$-values of a permutation test comparing the 5 incorrect models to the correct model, as a function of the number of available observations. Each model's deviation from the correct model is indicated in the panel above the plot. The dashed horizontal line highlights the classical $5\%$ significance level. For each curve, the leftmost point corresponds to the $p$-value for a single observation, $N=1$.
\label{pvalsPoisson}}
\end{figure}

To imitate common applications, we focus on 1 to 15 available observations. Specifically, we sample 100 mean scores for each of the models $F_1, .., F_6$ for $N = 1,...,15$ observed point patterns using the same sets of observations for all competing models. Each result for a wrong model is then compared to the corresponding result for the true model by a permutation test with 100 permutation resamples, yielding a total of 45,000 permutation tests (6 scoring rules, 5 competing models, 1 to 15 observations, 100 repetitions for each configuration). The average $p$-value for each configuration is reported in Figure~\ref{pvalsPoisson}.

The logarithmic score almost always has the highest probability of correctly identifying the true model, in line with intuition gathered from the Neyman-Pearson lemma suggesting that log-likelihood based inference tends to be most powerful, see e.g. the discussion in \citet{Lerch&2017}. Note, however, that the Neyman-Pearson lemma considers log-ratio tests and does not apply to permutation tests. The intensity scores with bandwidths $\sigma = 0.2,0.4,0.8$ lead to very similar $p$-values, slightly higher but comparable to the log-score $p$-values. Overall, the scores are quite good at identifying the correct model for 10 or more observations, except for the case where the model misspecification is only given by approximately 5\% more points. The overall high probabilities of all considered scores to identify the correct model under such small sample sizes is particularly promising considering the likeness of the competing models in this study.

\subsection{Bandwidth selection for the intensity score}

Generally, the choice of bandwidth for kernel estimators introduces a bias-variance trade-off; small bandwidths reduce the bias of the estimator but increase its variance \citep{Abramson1982}. Our findings here highlight this effect, since the scores with bandwidths $\sigma = 0.1$ and $\sigma = 1.6$ tend to have the lowest discriminability. The optimal bias-variance trade-off depends on mathematical properties of the underlying point process model \citep{DiggleMarron1988}, and is therefore unknown in practical applications.

The bandwidth parameter constitutes a penalty on spatial displacement of points: For larger bandwidths, the intensity estimator uses a wider smoothing kernel and the score becomes increasingly forgiving towards spatial displacements of points. This is highlighted in Figure~\ref{pvalsPoisson} by the lack of sensitivity towards the misspecification $\rho = 0.1$ for $\sigma = 1.6$. For the misspecification $\mu_x = 0.1$, the whole point pattern and thus the entire intensity estimate is shifted. This is detected by the score independent of bandwidth with a wider bandwidth having a variance-stabilizing effect, increasing the discriminability.  A small value of $\sigma$, on the other hand, increases the variance of the intensity estimate and thus decreases the sensitivity of the score.

In conclusion, the performance of the intensity score appears not overly sensitive to the choice of bandwidth. However, smaller bandwidths increase the sensitivity of the score to spatial displacements in the predictive model at the cost of higher sampling uncertainty in the score computation. We recommend choosing a small bandwidth, subject to it being substantially larger than the typical distance between nearest points. While the intensity score is almost as sensitive to misspecifications in the prediction as the logarithmic score, the logarithmic score may be preferred in situations where it can be applied to all competing models to avoid unnecessary computational costs and score uncertainty. 

\section{Earthquake rate prediction}\label{Earthquakes}

In this section we apply the intensity estimator score to evaluate earthquake rate predictions for California. Specifically, we focus on long-term predictions that quantify the expected rate of earthquakes over longer periods, e.g. five years. Such forecasts are, for example, used in urban planning \citep{Jena&2020}; for an overview, see \citet{Schorlemmer&2018} and references therein.

\subsection{Data}

Our data set is a catalog of all earthquakes of magnitude 3 or higher on the Richter scale in northern California--a window defined by $117.5^\circ -126^\circ$W and $34.5^\circ -42^\circ$N--that occurred between January 1, 1968 and December 31, 2019. The catalog contains a total of 17\,799 earthquakes and is publicly available at \url{http://www.ncedc.org}. The observation window and a subset of the data are shown in Figure~\ref{fig:EQexample}. 
 
\begin{figure}[!hbpt]
\centering
\includegraphics[width = 0.5\textwidth]{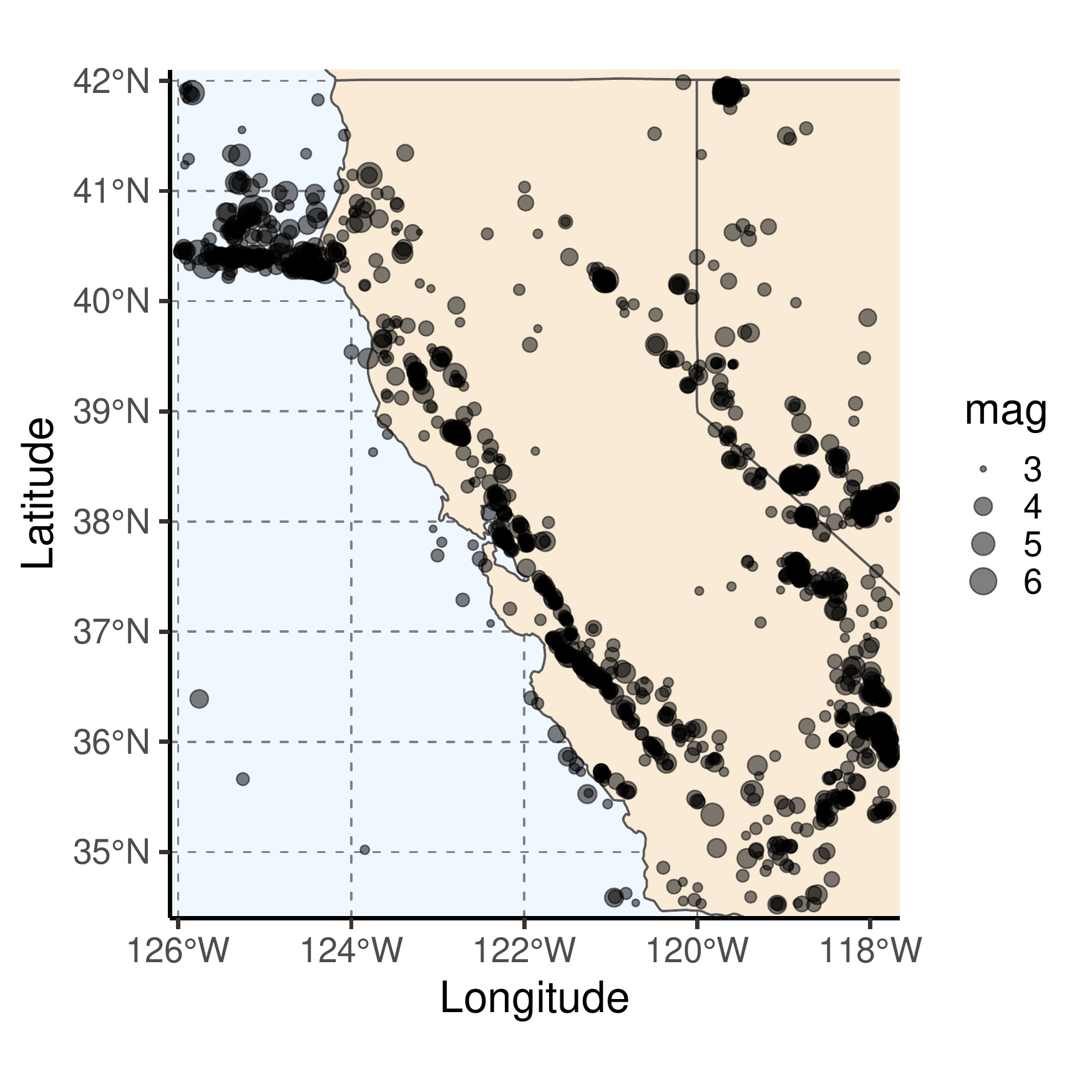}
\caption{Locations of earthquakes in northern California between 2011 and 2019. The point size indicates the magnitude of the earthquake on the Richter scale.\label{fig:EQexample}}
\end{figure}

\subsection{Competing models}

We follow \cite{VeenSchoenberg2006} and model the spatial distribution of earthquakes using a semi-parametric Poisson model with an intensity given by a mixture of an inhomogeneous intensity $\mu_\eta(x,y)$ and a homogeneous intensity $\nu$,
\[
\lambda_{\alpha,\eta} (x,y) = \alpha \mu_\eta(x,y) + (1-\alpha)\nu,
\]
where $\alpha\in [0,1]$ and $\eta >0 $ are model parameters. The inhomogeneous part $\mu_\eta(x,y)$ accounts for the fact that earthquake locations are generally clustered around seismic faults; it is estimated by kernel smoothing of all earthquakes in the training period with magnitude 3.5 or larger using a Gaussian kernel with standard deviation $\eta$ in both $x$- and $y$-coordinate and correlation estimated from the data. The homogeneous part $\nu$ accounts for earthquakes potentially occurring far away from previously recorded earthquakes; it is estimated as $N/A$ where $N$ is the total number of observed earthquakes with magnitude 3 or larger, and $A$ is the area of the observation window. The model assumes stationarity in time and the intensity is scaled by the length of the considered time period. 

The aim of the analysis is to find good values for the model parameters $\alpha$ and $\eta$ in a prediction context. In a cross-validation study, we compare spatial predictions for a range of values for $\alpha$ and $\eta$ under the intensity score and the logarithmic score.  \cite{VeenSchoenberg2006} previously performed an in-sample assessment of these models for earthquakes in southern California using a visual comparison of scaled $K$-functions for fixed $\eta=8$km and a range of values for $\alpha$. 

\subsection{Evaluation}

Assuming stationarity in time, we treat our data set as cyclic, i.e. we assume that the year 2019 is followed by the year 1968. A training set consists of 21 consecutive years, to match the set-up of \citet{VeenSchoenberg2006}. The remaining years (with a one year break before and after the training set to reduce correlation due to aftershocks) are used for evaluation, resulting in a validation set with 29 years of data. This procedure is repeated 52 times, so that each year of data is the first year in the training set exactly once.  

We compute mean scores for the parameter choices $\alpha\in \{0.6,0.65,...,1\}$ and $\eta\in\{2^k\,:\,k = 1,...,5\}$, where the unit of $\eta$ is kilometer. We employ the intensity score with bandwidths $\sigma = 4$ km and $\sigma = 8$ km in addition to the logarithmic score. The data contains on average 122 earthquakes within a 4 km radius of an earthquake location, and 277 earthquakes within an 8 km radius of an earthquake location, rendering these values substantially larger than the typcial distance between nearest points. 

\begin{figure}[h]
\includegraphics[width = 0.95\textwidth]{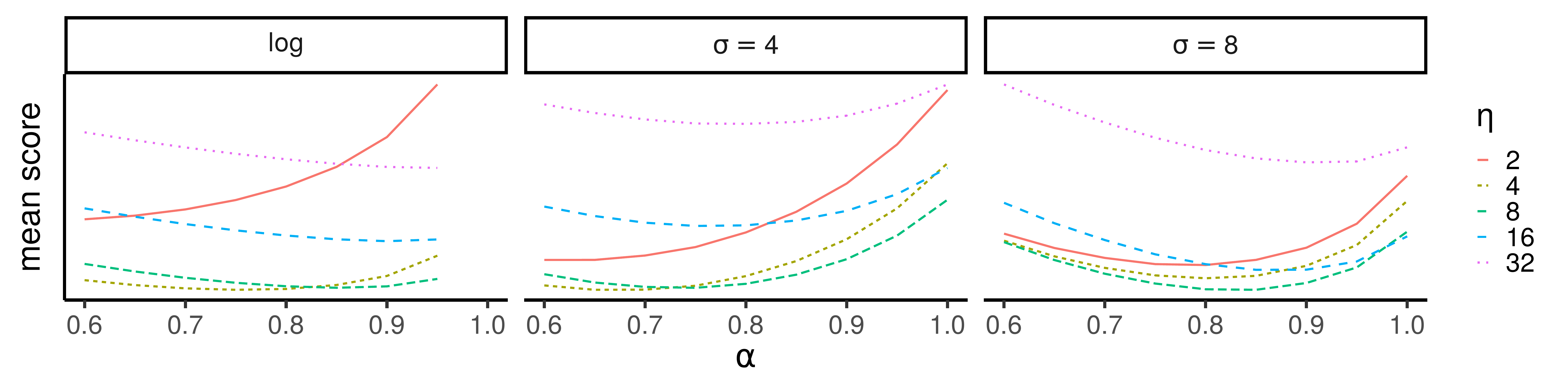}
\caption{Mean scores for earthquake rate prediction models based on a 52-fold cross-validation study, for the logarithmic score (left) and the intensity score with $\sigma=4$ km (middle) and $\sigma=8$ km (right). For ease of comparison, the scores have been rescaled to fit the same range.\label{fig:EQ_mean_scores}}
\end{figure}

The resulting scores are presented in Figure \ref{fig:EQ_mean_scores}. Overall, the model rankings resulting from the three scores show many similarities. In particular, all the scores agree that the best choices for $\eta$ are $4$ km or $8$ km, and the best choice of $\alpha$ is between $0.65$ and $0.85$. Specifially, the optimum under the logarithmic score is $(\alpha,\eta) = (0.75, 4)$, it is $(\alpha,\eta) = (0.65, 4)$ for the intensity score with $\sigma = 4$ km and $(\alpha,\eta) = (0.85, 8)$ for the intensity score with $\sigma=8$ km. The results indicate that the optimal value of $\alpha$ is increasing in $\eta$, indictating that the weight of the homogeneous component $\nu$ increases with a narrower kernel in the inhomogeneous component $\mu_\eta(x,y)$. Similarly, it can be observed that the intensity score with bandwidth $\sigma = 8$ km prefers both larger $\eta$ and larger $\alpha$ than the other two scores, in line with the intuition discussed in the previous section that increasing bandwidth corresponds to a lower penalty for spatial displacements of points.

\subsection{Discussion}

In their analysis, \cite{VeenSchoenberg2006} set $\eta = 8$ km and concluded, by visual inspection of inhomogeneous $L$-functions, that $\alpha\approx 0.7$ provided the best fit to their data.  This is in good agreement with our out-of-sample findings, even if the current data set covers a longer time period and we consider an  overlapping but not identical region.

We assume that the earthquake point process is stationary in time. However, since the data was collected over a timespan of more than 50 years, it is to be expected that the quality of the measurements varies over time. Indeed, the number of observations per year is more variable earlier in the time period, and exhibits a slight downward trend over time. While our study design partly accounts for this by merging early and late years, such features should be accounted for in operational prediction settings.

\section{Predicting spatial distribution of trees}\label{Trees}

In this application, we apply the $K$-function score to evaluate models that describe the spatial distribution of trees. The use of spatial point process models has substantially increased in ecology and forestry over the last decades, as they allow to address many key questions such as local dominance of species and species-area relationships, see \citet{Velazquez&2016} and \citet{Wiegand&2017} for overviews.

\subsection{Data}

We study location data of {\em Abies Amabilis} (Pacific silver fir)  at eight disjoint $25 \times 25$ m plots at Findley Lake Reserve in Washington State, USA. A description of the site conditions is given in \citet{Grier&1981}.  Figure~\ref{fig:silver fir four} shows the location of trees at two of the plots for three different time points over 31 years.  The area was clear-cut in 1957; the trees in our dataset were present as seedlings before the clear-cut and there appears to have been no reproduction in the stand since then.  The first observation was made in 1978, $21$ years after the area was clear-cut.  On average, roughly $80\%$ of the original trees were still present in the second observation in 1990 and approximately $25\%$ of them were present in the third observation in 2009. 

\begin{figure}[h]
\centering
\includegraphics[width=0.6\textwidth]{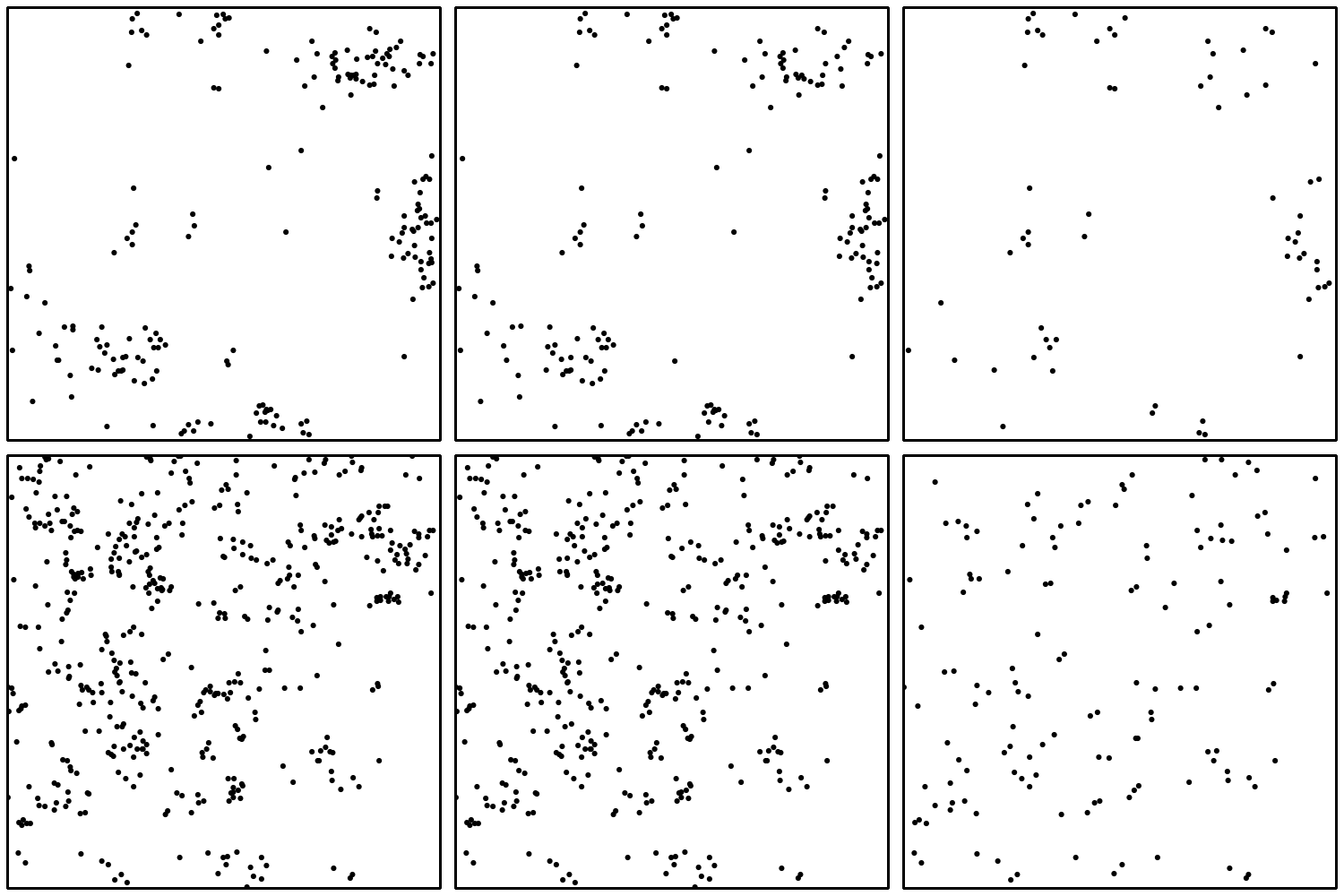}
\caption{{\em Abies amabilis} (Pacific silver fir) at two disjoint $25 \times 25$ m plots at Findley Lake Reserve in Washington State, USA.  The area was clear-cut in 1957; the first column shows trees present in 1978, 21 years after the clear-cut; the second column shows the trees still present in 1990, 33 years after the clear-cut; the third column shows the remaining trees in 2009, 52 years after the clear-cut.}
\label{fig:silver fir four}
\end{figure}

\subsection{Competing models}

Since the data do not show any signs of anisotrophy, we only consider spatially homogeneous models. Our emphasis lies in analyzing which model best captures the spatial clustering present in the data, particularly for the years 1978 and 1990 (see Figure~\ref{fig:silver fir four}). For this, we compare five competing models, all of which have been introduced in the context of modeling spatial distribution of plants, see the references given in Table \ref{CoxP}.  All five models belong to the class of doubly stochastic Poisson processes, or Cox processes, obtained by considering the intensity function itself a stochastic process. Four models are of the shot-noise type, where the random intensity is obtained by a sum of positive kernel functions $k(\cdot - x_i)$ centered around points $x_1,x_2,\ldots$ belonging to a homogeneous Poisson process with intensity $\kappa >0$. The models we consider are the Mat{\'e}rn, the Thomas, the variance Gamma and the Cauchy model. In addition, we consider the log Gaussian Cox process (LGCP). Table \ref{CoxP} gives a short characterization of each model, as well as references for more details.  There is a long-standing tradition in the literature of modelling spatial distribution of trees with the Mat\'ern and Thomas processes, cf. \cite{StoyanPenttinen2000}. More recently, the variance Gamma, the Cauchy and the LGCP models have become popular in this context, see e.g. \citet{Jalilian&2020}, \citet{Waagepetersen&2016} and \citet{Zhang&2019}.

\begin{table}[!hbpt]
  \caption{Specification of the five models used in the Pacific silver fir application. The first four models are shot-noise Cox processes, described in terms of a positive kernel $k$. They have an additional parameter $\kappa$, the intensity of the underlying homogeneous Poisson process. For the log Gaussian Cox process (LGCP), the random intensity is the exponential of a stationary Gaussian random field $Z$. Here, $b(0,\sigma)$ refers to the ball of radius $\sigma$ around 0, $\phi_\sigma$ is the isotropic mean-zero bivariate Gaussian kernel with standard deviation $\sigma I$, and $K_\nu$ denotes the modified Bessel function of the second kind.\label{CoxP} }
\footnotesize

\def\arraystretch{1.5}

\setlength{\tabcolsep}{12pt}
\begin{tabular}{llll}
  \toprule
   Model & Kernel function/intensity & Parameters & Reference\\
 \midrule
 Mat\'ern	&$k(w) = \xi  \mathbb{1}\{w \in b(0,\sigma)\}$	& $\kappa,\xi,\sigma$ &\citet{Matern1960}\\
 Thomas &$k(w) = \xi \phi_\sigma(w)$ & $\kappa,\xi,\sigma$	&\citet{Thomas1949}\\
 varGam & $k(w) = \xi\frac{1}{\pi2^{\nu + 1}\sigma^2\Gamma(\nu +1)}(\|w\|/\sigma)^\nu K_\nu(\|u\|/\sigma)$		& $\kappa,\xi,\sigma,\nu$ &\citet{Jalilian2013}\\
 Cauchy & $k(w) = \frac{1}{2\pi\sigma^2}\bigg(1 + \frac{\|w\|^2}{\sigma^2}\bigg)^{-3/2}$ & $\kappa,\xi,\sigma$ & \citet{Ghorbani2013}\\
 LGCP &$\lambda(w) = \exp(Z(w)),$ $\E[Z(w)] = \mu,$ & $\mu,\tau,\sigma$	&\citet{Moeller1998}\\
 & $\text{Cov}(Z(w_1),Z(w_2)) = \tau^2 \exp(-\|w_1-w_2\|/\sigma)$ && \\
\bottomrule
\end{tabular}
\end{table}

All five models are fitted in the same manner by minimum contrast, see \citet{DiggleGratton1984} and \citet{Waagepetersen2007}. First, the empirical $K$-function is estimated from the training set by averaging the estimator in \eqref{eq:K hat} across all plots in the training set. Second, the parameters $\Theta$ of the model are fitted to minimize the integrated distance
\begin{align}\label{mincon}
\wh \Theta = \argmin_\Theta\bigg\{\int_0^{r_{\max}} (\wh K(r)^{1/4} - K(r;\Theta)^{1/4})^2\,dr\bigg\},
\end{align}
where $K(r;\Theta)$ is the theoretical $K$-function for the model with parameter $\Theta.$ The upper limit $r_{\max}$ should be chosen small relative to the plot dimensions \citep{Diggle2013}, but the choice of $r_{\max}$ introduces a certain level of arbitrariness in the estimation procedure. We choose $r_{\max} = 6.25,$ one fourth of the plot length, the default value of \texttt{spatstat}.

\subsection{Evaluation}

The observations for each year are analyzed separately such that in each analysis, eight independent realizations of the underlying process are available. We consider a prediction scenario where four of the observations are used to fit a predictive distribution, which is then verified against the other four observations. This is repeated for all possible combinations, leading to a total of 280 score evaluations per year. Some of the scores are correlated, namely if they are based on the same observation (but use potentially different training sets), or if their training sets overlap. 

\begin{figure}[!hbpt]
\centering
\begin{subfigure}[t]{0.32\textwidth}
\includegraphics[width = \textwidth]{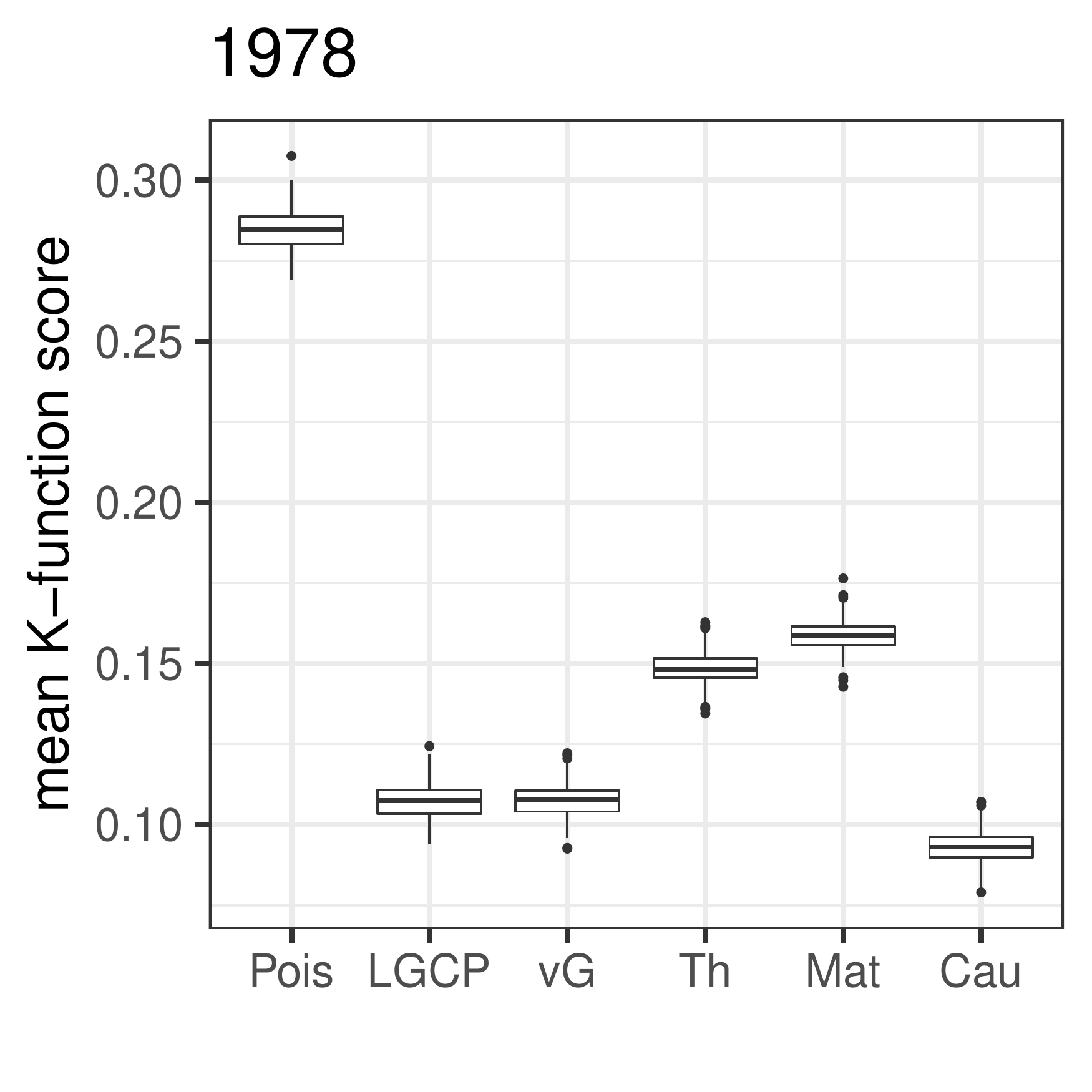}
\end{subfigure}
\begin{subfigure}[t]{0.32\textwidth}
\includegraphics[width = \textwidth]{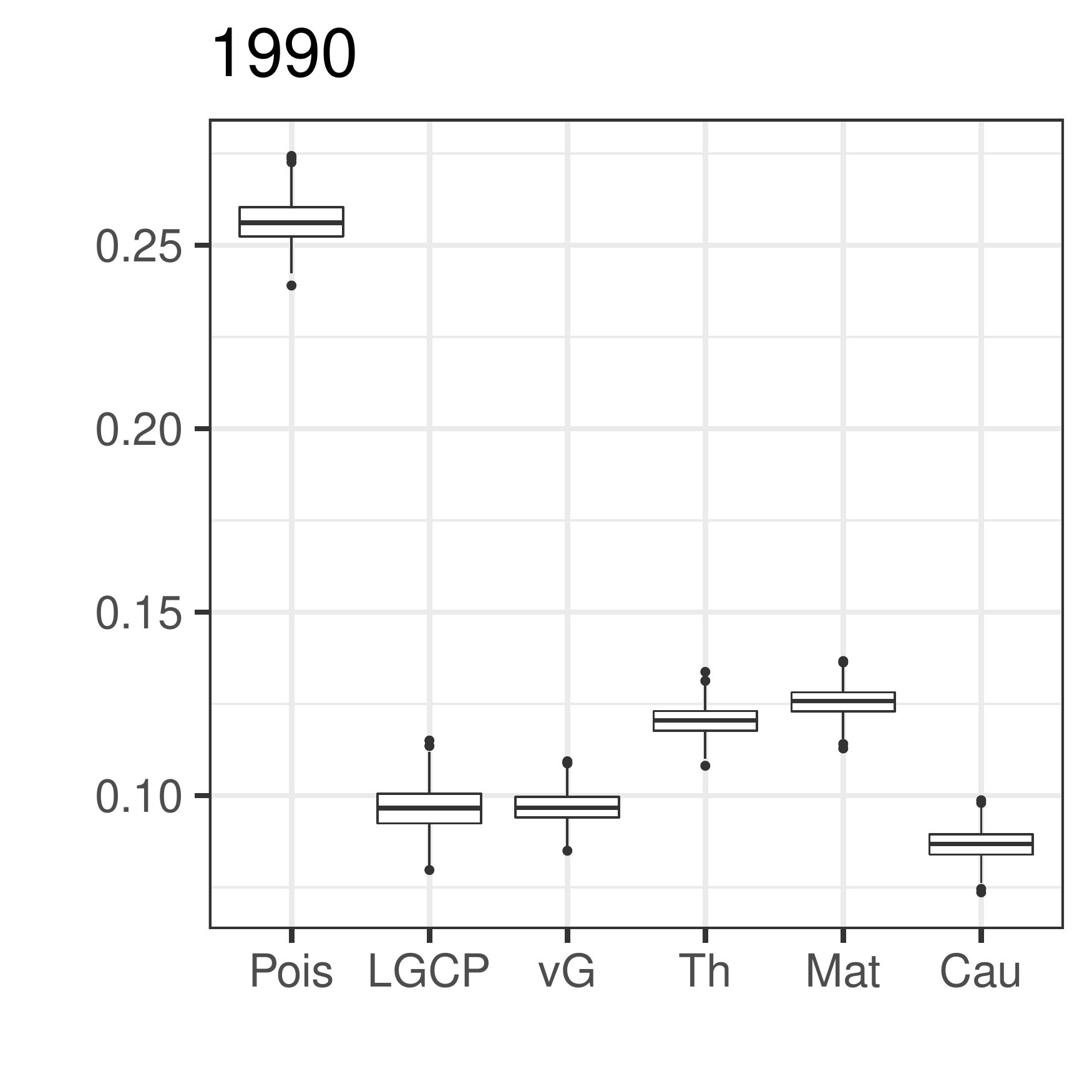}
\end{subfigure}
\begin{subfigure}[t]{0.32\textwidth}
\includegraphics[width = \textwidth]{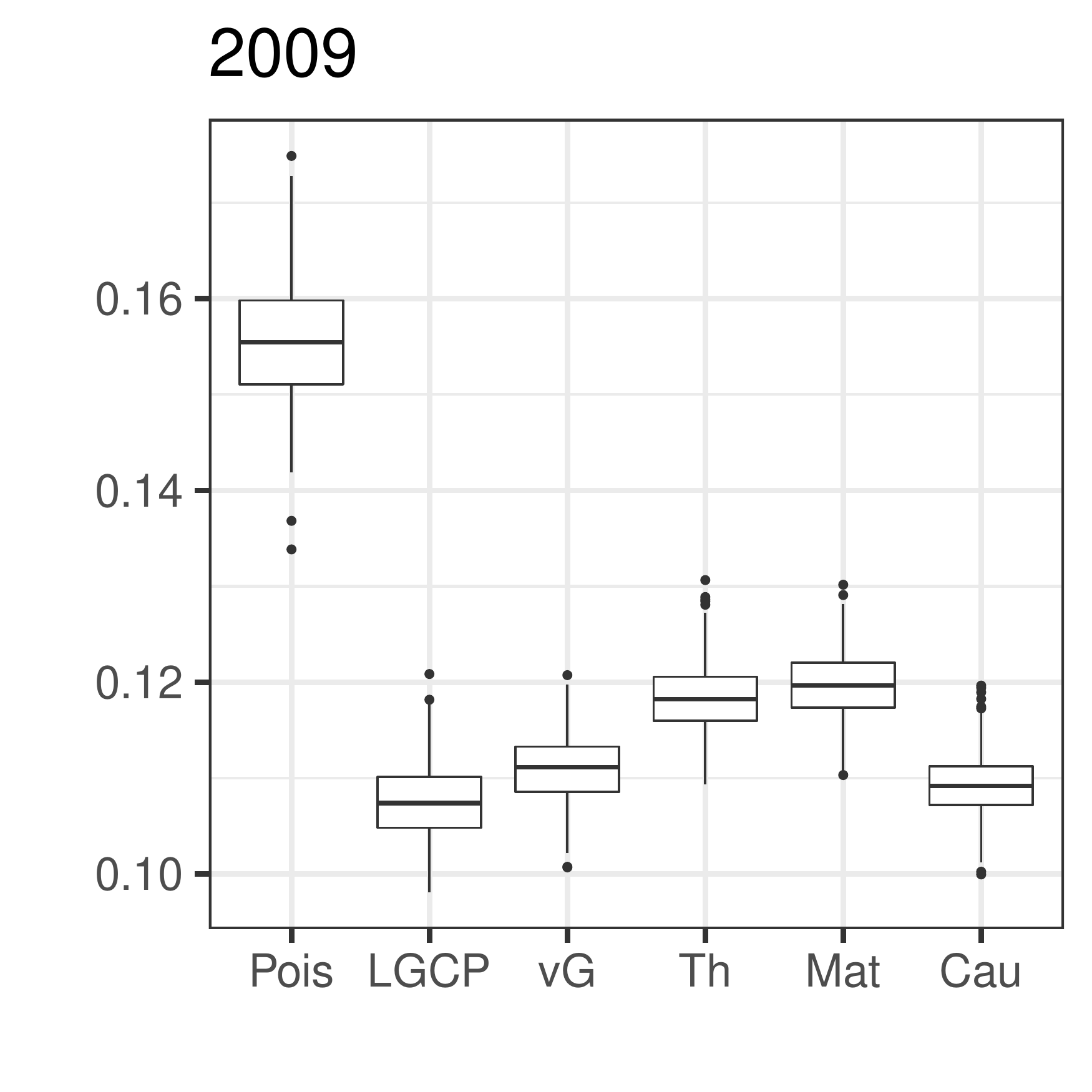}
\end{subfigure}
\caption{
Mean $K$-function scores for predicting the spatial distribution of Pacific silver fir for the three years available in the dataset.
The models are specified in Table \ref{CoxP} and a homogeneous Poisson model is shown for reference. The boxplots indicate bootstrapped distributions of the mean $K$-function score.\label{Fig:Treescores}}
\end{figure}

With a focus on the appropriate modelling of point interaction, we employ the $K$-function score for evaluating predictive performance. The results are shown in Figure \ref{Fig:Treescores}. For all three time points, the LGCP, variance Gamma and Cauchy model clearly outperform the Mat\'ern and Thomas model. The overall best score is achieved by the Cauchy model for the years 1978 and 1990, and by the LGCP for the year 2009. For the data from 1978 and 1990, permutation tests confirm the superior performance of the Cauchy model against all competing models ($p$-values below 0.01 for all models). For these years, the LGCP mean score is slightly lower than that for the variance Gamma process. These differences, however, are not significant ($p$-value of $0.4$ for 1978 and of $0.46$ for 1990). For the 2009 data, the LGCP model does not significantly outperform the Cauchy model ($p$-value of 0.15), but it is significantly better than the variance Gamma model ($p$-value of 0.04). The lower significance and higher variability of scores for 2009 is caused by the smaller number of observed trees per plot. A reduced number of points increases the variance in the $K$-function estimator, increasing the score uncertainty.

\subsection{Discussion}

Previous studies have compared the fit of several of these models to related data sets. \citet{Ghorbani2013} analyzes long-leaf pine location data and shows that the Cauchy model provides a better fit than the Thomas model. \citet{Zhang&2019} fit both the Mat\'ern and the variance Gamma model to location data of black locust trees and conclude that the variance Gamma model provides a better fit. In \citet{Jalilian2013}, the authors fit the Thomas, the variance Gamma and the Cauchy model to data sets of three species of rainforest trees. The variance Gamma model provides the best fit for two of the three considered species, and the Cauchy model for the third. Both the Cauchy and the variance Gamma model consistently outperform the Thomas model for all three species. In these studies, the model assessment is performed in-sample. Our out-of-sample  study mostly leads to similar conclusions, except that the variance Gamma model gets outperformed by the Cauchy and LGCP models.
A possible reason for this is that the variance Gamma model has an additional parameter compared to the other models.
This additional flexibility gives an advantage in in-sample studies where the model fit is assessed. In a prediction setting, on the other hand, selecting models with a high number of parameters can increase the risk of overfitting, resulting in reduced predictive performance.

\section{Conclusions and discussion}\label{Discussion}

This paper introduces a new class of proper scoring rules that combines estimators for summary statistics with the continuous ranked probability score (CRPS). The proposed scoring rules can be computed from random draws of the predictive models. Therefore they can be applied to a wider range of predictive distributions than the commonly used logarithmic score, which requires the density of the predictive model to be known. Moreover, the $K$-function score is, to the best of our knowledge, the first proper score for point processes focusing on correct representation of point interactions.

Generally, our scores do not require knowledge of mathematical properties of the predictive models, just the ability to generate random draws. This is a particular advantage in the context of point processes where mathematical properties are frequently untractable. The \texttt{R} package \texttt{spatstat} \cite{spatstat} provides a diverse set of tools for simulating a large set of spatial point process models and provides implementations of all common summary statistic estimators. \texttt{spatstat} therefore provides a platform that makes the introduced scoring rules easily applicable in a variety of situations. The code for all studies presented in this paper is publicly available at \url{github.com/ClaudioHeinrich/PointProcessSRs}.

Our approach is based on the intuitive principle that, when the observation space is complex, the observations and predictions can be mapped into a simpler space for validation. This approach is not restricted to point processes, and opens a fruitful new perspective on the validation of forecasts that live on complex spaces. Indeed, when the observation space is involved, finding proper scoring rules can be difficult, even more so when they ought to be sensitive to certain high level properties of the observation-generating process. The mapping approach replaces this challenge by the much easier task of finding real-(or function-)valued mappings sensitive to these properties. Potential other applications include high-dimensional forecasts and forecasts of spatial fields, as well as function-valued forecasts.

In the context of point processes, it is natural to use estimators of summary statistics in the mapping principle, as they are designed to be sensitive to high-level properties of the point process such as clustering or inhomogeneity. The variogram score from \citet{ScheuererHamill2015} follows a similar principle in a different setting, utilizing the variogram estimator of random vectors. An additional advantage of this approach is that there is a wide range of literature analyzing summary statistics and their estimators, and practicioners are familiar using them, which makes scoring rules based on summary statistics easier to interpret.

Besides the subjective selection of  a mapping, the mapping approach requires the selection of a proper scoring rule on the codomain. We opted for the CRPS since it is \emph{strictly} proper and can be efficiently approximated by Monte-Carlo methods. 
When computational time is highly important, an attractive alternative might be the mean square error (MSE), $S(\bb y, F) := (\bb y - \E_F[\bb X])^2$. The expectation of the pushforward measure $\wh T(F)$ can be Monte-Carlo approximated, and the computational costs are lower, since for the CRPS the term $\E_F[|\bb y - \bb X|]$ needs to be approximated for each observation $\bb y$ separately. However, since the MSE is not strictly proper the resulting score may be less sensitive to miscalibrations.

There are several avenues for potential future research. Generally, assessing the usefulness of this mapping approach in other contexts as mentioned above should be investigated. An important application for point-process-valued forecasts is earthquake rate forecasting, and a variety of methods for forecast validation has been developed in this community. However, prediction so far mostly focused exclusively on the intensity, and moving to probabilistic predictions is an ongoing endeavor, see \cite{Schorlemmer&2018}. This requires the introduction of new evaluation metrics, and the scores presented in this paper could provide an entry point.


\end{document}